\documentclass[aps,pra,10pt,twocolumn,notitlepage,superscriptaddress]{revtex4-2}

\usepackage{algorithm,algpseudocode,amsfonts,amsmath,amssymb,amsthm,array,dsfont,float,graphicx,hyphenat,lmodern,mathtools,microtype,multirow,tikz,xcolor}
\usepackage[UKenglish]{babel}
\usepackage[T1]{fontenc}
\usepackage[scaled=0.8]{FiraMono}
\usepackage[utf8]{inputenc}
\usepackage[UKenglish]{isodate}
\usepackage[caption=false]{subfig}
\usepackage{hyperref}
\usepackage[capitalize,nameinlink]{cleveref}

\definecolor{darkred}{rgb}{0.5,0,0}
\definecolor{darkgreen}{rgb}{0,0.5,0}
\definecolor{darkblue}{rgb}{0,0,0.5}
\hypersetup{colorlinks,breaklinks,linkcolor=darkred,citecolor=darkgreen,urlcolor=darkblue}

\DeclareMathOperator{\conv}{conv}
\DeclareMathOperator{\diam}{diam}
\DeclareMathOperator{\ext}{ext}
\DeclareMathOperator{\vol}{vol}
\DeclareMathOperator{\argmax}{argmax}
\DeclareMathOperator{\argmin}{argmin}

\renewcommand{\leq}{\leqslant}
\renewcommand{\geq}{\geqslant}
\renewcommand{\arraystretch}{1.4}

\newtheorem{theorem}{Theorem}
\newtheorem{proposition}[theorem]{Proposition}
\newtheorem{definition}[theorem]{Definition}

\newtheorem{observation}[theorem]{Observation}
\newtheorem{corollary}[theorem]{Corollary}
\newtheorem{example}[theorem]{Example}

\newcommand{\ket}[1]{|#1\rangle}
\newcommand{\ketbra}[1]{|#1\rangle \langle#1|}

\graphicspath{{./fig/}}
\newcommand{\bra}[1]{\mbox{$ \langle #1 | $}}

\newcommand{\tr}{\mathrm{tr}}

\newcommand{\x}{\mathbf{x}}

\begin{document}

\author{Ye-Chao Liu}
\email{ye-chao-liu@hainanu.edu.cn}
\affiliation{School of Physics and Optoelectronic Engineering, Hainan University, Haikou, 570228, China}
\affiliation{Zuse Institute Berlin, Takustraße 7, 14195 Berlin, Germany}

\author{Jannis Halbey}
\affiliation{Zuse Institute Berlin, Takustraße 7, 14195 Berlin, Germany}
\affiliation{Institut für Mathematik, Technische Universität Berlin, Straße des 17. Juni 136, 10623 Berlin, Germany}

\author{Sebastian Pokutta}
\affiliation{Zuse Institute Berlin, Takustraße 7, 14195 Berlin, Germany}
\affiliation{Institut für Mathematik, Technische Universität Berlin, Straße des 17. Juni 136, 10623 Berlin, Germany}

\author{Sébastien Designolle}
\email{sebastien.designolle@inria.fr}
\affiliation{Zuse Institute Berlin, Takustraße 7, 14195 Berlin, Germany}
\affiliation{Inria, ENS de Lyon, UCBL, LIP, 69342, Lyon Cedex 07, France}

\title{Unified Toolbox for Multipartite Entanglement Certification}
\date{rd November 2026}

\begin{abstract}
We present a unified framework for multipartite entanglement characterization based on the conditional gradient (CG) method, incorporating both fast heuristic detection and rigorous witness construction with numerical error control.
Our method enables entanglement certification in quantum systems of up to ten qubits and applies to arbitrary entanglement structures.
We demonstrate its power by closing the gap between entanglement and separability bounds in white noise robustness benchmarks for a class of bound entangled states.
Furthermore, the framework extends to entanglement robustness under general quantum noise channels, providing accurate thresholds in cases beyond the reach of previous algorithmic methods.
These results position CG methods as a powerful tool for practical and scalable entanglement analysis in realistic experimental settings.
\end{abstract}

\maketitle

\textit{Introduction.---}
Entanglement~\cite{HHHH09} is a fundamental feature of quantum mechanics that distinguishes it from classical physics, and serves as a central resource in quantum technologies such as communication, computation, metrology, and teleportation.
Accordingly, the detection and characterization of entanglement remain pivotal tasks in quantum information science~\cite{GT09}.
For bipartite systems, significant progress has been made through various entanglement criteria.
The Peres-Horodecki criterion, also known as the positive partial transpose (PPT) criterion, is a well-known example~\cite{Per96, Horo97}.
Another one is the computable cross norm or realignment (CCNR) criterion, which is more experimentally accessible and can detect certain bound entangled states that the PPT criterion cannot~\cite{CW03, Rud05}.

When it comes to multipartite systems, the challenges become significantly greater.
Not only does the Hilbert space grow exponentially, but the multipartite entanglement also has much richer structures, ranging from fully separable to genuine multipartite entanglement (GME).
While both PPT and CCNR criteria have been extended to detect GME~\cite{NMG13}, entanglement witnesses remain the most widely used tool for multipartite systems~\cite{HHH96, Ter00, LKCH00, BCH+02}.
An entanglement witness $\mathcal{W}$ is a necessary and sufficient criterion satisfying $\tr(\mathcal{W}\rho)\geq 0$ for all separable $\rho$, and $\tr(\mathcal{W}\rho)< 0$ for at least one entangled $\rho$.
They can be analytically constructed for arbitrary pure states~\cite{BEK+04} but often fall short when dealing with general (mixed) entangled states, including bound entangled states.

To overcome these limitations, algorithmic methods have been developed to address arbitrary entangled states.
The PPT criterion can be formulated as a semidefinite program (SDP), and its symmetric extensions, known as the Doherty–Parrilo–Spedalieri (DPS) hierarchy, can asymptotically approximate the separable space~\cite{DPS04, NOP09}.
On the other hand, by inner polytope approximation of the Hilbert space, there are also methods based on SDP that can certify separability~\cite{OYGN24}.
State-of-the-art SDP-based methods have been demonstrated for systems up to six qubits, yet they remain inadequate for current experimental regimes involving more than ten qubits~\cite{SXL+17}, where computational breakdowns often occur.

Beyond SDP-based methods, conditional gradient (CG) methods have emerged as a promising alternative for multipartite entanglement characterization.
The idea was first explored by Kampermann \textit{et al.}~\cite{KGWB12}, and later developed into suitable algorithms within CG methods~\cite{SG18, WPSW20, PSW20, Hu_2023_algorithm}.
A key advantage of CG methods is their flexibility in handling arbitrary multipartite entanglement structures.
However, existing CG-based techniques have primarily focused on separability certification and are constrained to small systems as well.

\begin{figure}[bt]
  \includegraphics[width=\columnwidth]{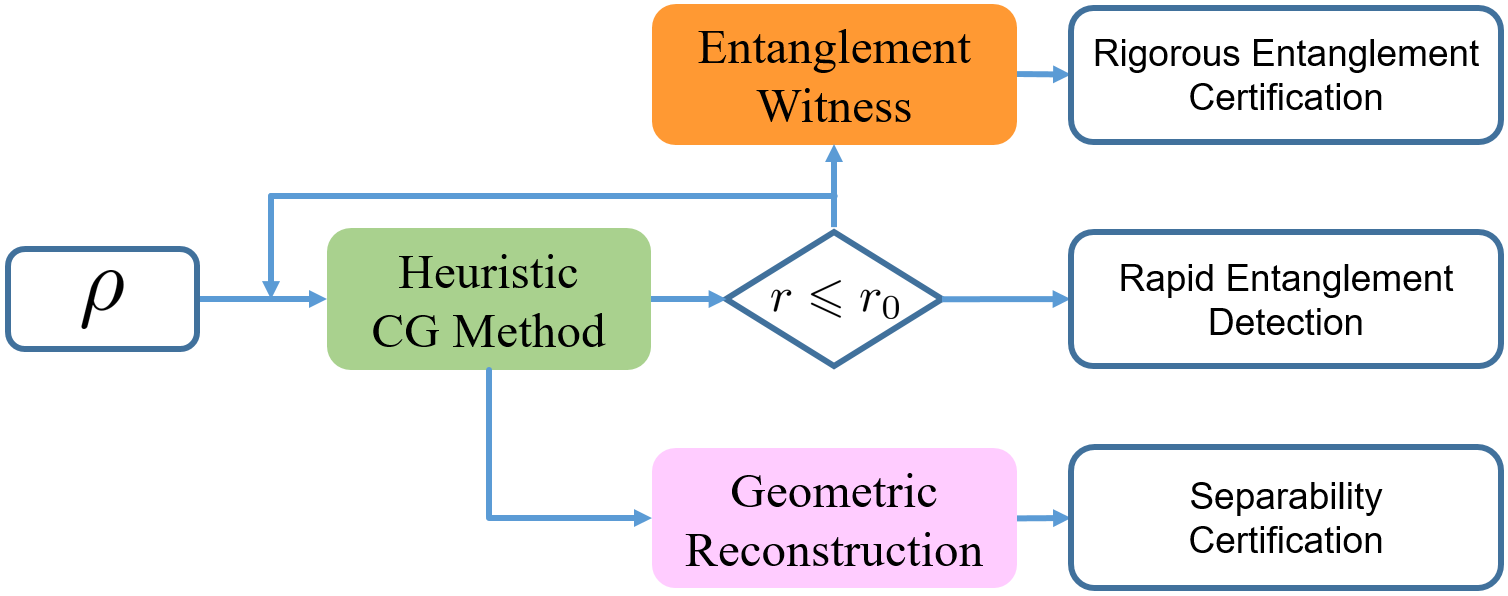}
    \caption{
        Schematic overview of our unified toolbox for multipartite entanglement certification based on conditional gradient (CG) methods.
        Starting from a target state $\rho$, the heuristic CG iteration evaluates the dual gap efficiency $r$ to determine the appropriate certification path.
        If $r \leq r_0$, entanglement is rapidly detected; otherwise, further CG iterations are performed or a rigorous entanglement witness is constructed to certify entanglement while accounting for numerical error.
        Alternatively, geometric reconstruction from the same CG-based output can be used to certify separability.
    }
    \label{fig:framework}
\end{figure}

In this work, we introduce a CG-based framework capable of certifying entanglement for multipartite quantum systems.
Our approach significantly extends the scalability of CG methods, enabling rapid entanglement detection in systems of up to ten qubits.
Moreover, it yields rigorous entanglement witnesses suitable for experimental validation.
Combined with separability certification via the same CG-based framework, our method constitutes a unified toolbox for multipartite entanglement analysis.
We demonstrate the power of our method by providing optimal entanglement thresholds on the entanglement robustness for a class of bound entangled states, thereby closing the gap between entanglement and separability thresholds.
Last but not least, our toolbox can assess robustness under nonlinear noise models and realistic noise channels, making it broadly applicable in experimental scenarios.
As illustrated in \cref{fig:framework}, our unified toolbox integrates fast detection, rigorous certification, and separability analysis into a coherent structure.

\textit{Entanglement and Frank-Wolfe algorithm.---}
We begin by recalling the definition of multipartite entanglement.
A pure state is called $k$-separable if it can be written as $\ket{\psi^{k\rm s}} = \ket{\phi_1} \otimes \ket{\phi_2} \otimes \cdots \otimes \ket{\phi_k}$, where each $\ket{\phi_i}$ may itself be entangled within a subset of the total system.
For a mixed state, it is $k$-separable if it can be expressed as a convex combination of $k$-separable pure states: $\sigma^{k\rm s} = \sum_i \lambda_i \ket{\psi^{k\rm s}_i}\bra{\psi^{k\rm s}_i}$, where the coefficients $\lambda_i$ form a probability distribution.
Usually, if $n$-partite quantum states are not $n$-separable, then they are called entangled states.
And non-$2$-separable states are called GME.

Accordingly, the set of all $k$-separable states forms a convex subset $\mathcal{S}_{k}$ of the state space, and any state outside $\mathcal{S}_{k}$ is entangled.
Characterizing multipartite entanglement then amounts to determining whether a given quantum state $\rho$ belongs to $\mathcal{S}_{k}$.
This naturally leads to the following optimization problem:
\begin{equation}\label{eq:opt}
    f(\rho) = \min_{\sigma\in\mathcal{S}_k} f(\rho, \sigma)\,,
\end{equation}
where $f(\rho, \sigma)$ denotes a distance measure between quantum states.
If $f(\rho) = 0$, the state is $k$-separable; otherwise, it is entangled.
A common choice for the distance measure is the distance induced by the Hilbert-Schmidt norm $\|\rho - \sigma\|$, also known as the Frobenius norm or $l_2$
norm.
For algorithmic convenience, we consider half of its squared form $f(\rho, \sigma) = \frac{1}{2} \|\rho - \sigma\|^2$, which is convex and differentiable.
Although not a strict distance, we will refer to $f(\rho)$ as a ``distance'' throughout, reflecting its geometric role in the optimization.

This is a constrained convex optimization problem, and the typical algorithmic tool is the Frank–Wolfe method, also known as the conditional gradient method.
The basic procedure for entanglement characterization via the vanilla Frank-Wolfe algorithm is as follows:
\begin{algorithm}[H]
  \renewcommand\thealgorithm{}
  \floatname{algorithm}{}
  \caption{Vanilla Frank–Wolfe algorithm}
  \label{alg:FW}
  \begin{algorithmic}[1]
    \Require Distance function $f(x)$; initial point $\sigma_0 \in \mathcal{S}_k$
    \For{$t = 0, \ldots, m$}
      \State $\psi_t \gets \arg\min_{\psi \in \mathcal{S}_k} \bra{\psi} \sigma_t - \rho \ket{\psi}$
      \State $\gamma_t \gets \arg\min_\gamma f(\sigma_t + \gamma(\psi_t - \sigma_t))$
      \State $\sigma_{t+1} \gets \sigma_t + \gamma_t(\psi_t - \sigma_t)$
    \EndFor
    \Ensure $\sigma_1, \ldots, \sigma_m \in \mathcal{S}_k$
  \end{algorithmic}
\end{algorithm}

Each iteration constructs a new state $\sigma_{t+1}$ such that $f(\rho, \sigma_{t+1}) \leq f(\rho, \sigma_t)$.
The product state $\psi_t$, obtained via a linear minimization oracle (LMO), corresponds to the extreme point in $\mathcal{S}_k$ that is most aligned with the descent direction $\sigma_t - \rho$, i.e., it minimizes the directional derivative of the distance function along this direction.
This makes the update efficient and steers $\sigma_t$ progressively toward $\rho$.
An essential advantage of Frank–Wolfe algorithms over other gradient-based approaches is to avoid costly projections onto $\mathcal{S}_k$, which are as difficult as the problem~\eqref{eq:opt} we are trying to solve, which is itself known to be NP-hard~\cite{Gur03}.

A related idea was previously introduced in the work of Kampermann \textit{et al.}~\cite{KGWB12}, and later studies have adopted variants based on Gilbert’s algorithm~\cite{SG18, WPSW20, PSW20, Hu_2023_algorithm}, which provides guaranteed convergence, in contrast to the earlier approach by Kampermann \textit{et al}.
Although Gilbert's algorithm~\cite{Gil66} was used earlier in the quantum information literature, it is a special case of the more general Frank-Wolfe method~\cite{FW56, LP66}, which was introduced earlier and has since seen significant advances in algorithmic theory.
In this work, we adopt the term \textit{conditional gradient} (CG) method to unify terminology, encompassing both the vanilla Frank–Wolfe (or Gilbert's) algorithm and its modern improvements \cite{wolfe1976normpoint, guelat1986some, holloway1974extension, tsuji2022pairwise, pok17lazy, halbey2025efficientquadraticcorrectionsfrankwolfe}.
The modern CG variants employed here offer a provable and substantial improvement in time complexity, enabling rapid entanglement detection in larger quantum systems. Further algorithmic details can be found in Supplemental Material \cite{supp} and in recent review articles~\cite{BRZ21, BCC+22}.
Our code is available as a Julia package~\footnote{\label{code}See our code as a Julia package at \url{https://github.com/ZIB-IOL/EntanglementDetection.jl}.}.

\textit{Rapid entanglement detection.---}
By avoiding explicit projections onto the separable space, the main computational cost of our method lies in evaluating the linear minimization oracle (LMO), which can be written as
\begin{equation}\label{eq:heuristic}
    \min_{\ket{\phi_i} \in \mathcal{H}_i, i=1,\cdots,k}
    \bra{\phi_1} \otimes \cdots \otimes \bra{\phi_k}
    (\sigma_t - \rho)
    \ket{\phi_1} \otimes \cdots \otimes \ket{\phi_k}\,.
\end{equation}
This optimization can be heuristically solved via alternating updates over subsystems.
However, due to numerical error and the inherent limitations of finite iterations, the computed quantity $f(\rho, \sigma_t)$ generally does not converge to zero, even for separable states.
Thus, $f(\rho, \sigma_t)$ cannot be directly used to certify entanglement or separability.
In previous works~\cite{SG18}, a geometric reconstruction approach was developed within the CG-based framework for separability certification.

Here, we take the complementary perspective to utilize CG methods for certifying entanglement.
Our approach is based on the following proposition:
\begin{proposition}\label{prop:ent}
A quantum state $\rho$ is entangled if and only if there exists a separable state $\sigma$ such that $\|\sigma - \tau\| < \|\rho - \tau\|$ for all $\tau \in \mathcal{S}_k$.
\end{proposition}
\begin{proof}
Necessity.---
Let $\sigma = \arg\min_{\sigma \in \mathcal{S}_k} \|\rho - \sigma\|$.
For any $\tau \in \mathcal{S}_k$, we have $\tr[(\sigma - \rho)(\sigma - \tau)] \leq 0$, and therefore we achieve $\|\rho - \tau\|^2
> \|\rho - \sigma\|^2$ following the parallelogram identity.
Sufficiency.---
If $\rho$ is separable, let $\tau = \rho \in \mathcal{S}_k$.
Then the assumption $\|\sigma - \rho\| < \|\rho - \rho\| = 0$ leads to a contradiction.
\end{proof}

We now introduce the dual gap $g_t = \tr[(\sigma_t - \rho)(\sigma_t - \psi_t)]$, a standard convergence indicator in CG methods~\cite{BCC+22}.
With the parallelogram identity, we can rewrite it as
\begin{equation}
g_t = \tfrac{1}{2}\|\rho - \sigma_t\|^2 + \tfrac{1}{2}\|\sigma_t - \psi_t\|^2 - \tfrac{1}{2}\|\rho - \psi_t\|^2\,.
\end{equation}
This motivates the following entanglement criterion:
\begin{corollary}\label{coro:criterion}
A quantum state $\rho$ is entangled if and only if the inequality $f(\rho, \sigma_t) > g_t$ holds.
\end{corollary}

Therefore, instead of simply checking whether $f(\rho, \sigma_t) > 0$, we establish a practical entanglement detection criterion by defining the dual gap efficiency
\begin{equation}
r_t := \frac{g_t}{f(\rho, \sigma_t)}\,,
\end{equation}
which compares the dual gap $g_t$ with the current distance $f(\rho, \sigma_t)$.
The algorithm terminates and declares $\rho$ entangled once $r_t < 1$.
As shown in \cref{fig:shortcut}(a), the three-qubit Greenberger-Horne-Zeilinger (GHZ) state with $80\%$ white noise lies exactly on the boundary of $\mathcal{S}_3$, resulting in a nearly constant $r_t$ throughout iterations.
When the noise is reduced by $1\%$, introducing slight entanglement, the criterion is satisfied after a few thousand iterations.
States with higher entanglement (i.e., less noise) are certified more rapidly.

This capability enables entanglement detection in large systems, particularly in practical scenarios.
In experiments, entangled states often contain moderate noise, yet larger systems typically tolerate more.
To illustrate this, \cref{fig:shortcut}(b) shows that $10$-qubit GHZ and Dicke states mixed with $70\%$ white noise can still be certified within dozens of iterations.
The total runtime is under one minute using a 2.80~GHz CPU and 16~GB of memory.
In the main text, we focus on detecting full separability or bipartite entanglement (i.e., $\mathcal{S}_k$ with $k=n$ for $n$-qubit systems); for general entanglement structures, see Appendix B \cite{supp}, where we improve upon existing results \cite{Guhne_2010_separability, Gao_2010_detection, Jungnitsch_2011_taming, Ananth_2016_nonkseparability, Chen_2018_necessary, SG18, Ge_2021_tripartite} in most cases.

\begin{figure}[t]
  \includegraphics[width=\columnwidth]{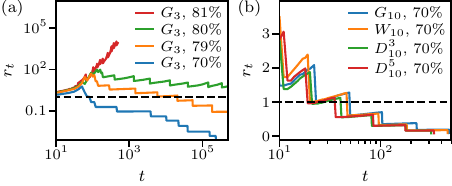}
    \caption{
      The ratio $r_t=g_t/f(\rho,\sigma_t)$ for (a) three-qubit GHZ state mixed with varying levels of white noise, and (b) ten-qubit GHZ and Dicke states with $70\%$ white noise.
      The dashed line denotes the threshold $r=1$ for entanglement detection.
    }
    \label{fig:shortcut}
\end{figure}

Our results reveal a clear connection between the amount of entanglement and the computational effort required.
For a fixed system size, more entangled states require fewer iterations to satisfy the criterion.
For instance, the three-qubit GHZ state with less noise converges faster; the Dicke state outperforms the GHZ state in the $10$-qubit case, as it exhibits stronger robustness to white noise.

However, this property also implies that detecting weak entanglement near the boundary of $\mathcal{S}_k$ may require prohibitively many iterations.
Moreover, since the LMO problem~\eqref{eq:heuristic} is solved heuristically, there is no guarantee of global optimality.
This compromises the reliability of \cref{coro:criterion}, which critically depends on the accuracy of $\psi_t$ and $g_t$.
In practice, this can be mitigated by tightening the threshold, e.g., accepting $r_t < 1/5$, which suffices in most cases, although without theoretical guarantee.

In order to strictly avoid the above two possible problems, we introduce in the following an entanglement witness, which explicitly accounts for numerical uncertainty and therefore can certify entanglement rigorously.
It can be seen as a subroutine within the same algorithmic framework.

\textit{Rigorous entanglement witness.---}
An entanglement witness for a generic quantum state can be constructed as~\cite{PR01, BNT02, Bertlmann2005optimal}
\begin{equation}
    \mathcal{W} = \sigma - \rho_e + \operatorname{Tr}[\sigma(\rho_e - \sigma)] \cdot \openone\,,
\end{equation}
where $\rho_e$ is the target entangled state to be certified, and $\sigma$ is the closest separable state in $\mathcal{S}_k$ (normalization omitted).
The main difficulty lies in finding $\sigma$, which corresponds exactly to the solution of \cref{eq:opt}.

Instead of the heuristic LMO procedure, we now consider a structured search over an $\varepsilon$-net $\mathcal{S}_k^\varepsilon$ of the boundary of the separable set $\mathcal{S}_k$ \cite{Pis89, AS17, EG00}.
It is a finite set of pure separable states satisfying the condition that for every $\sigma$ in the boundary of $\mathcal{S}_k$, there exists $\sigma^\varepsilon \in \mathcal{S}_k^\varepsilon$ such that $\|\sigma - \sigma^\varepsilon\| \leq \varepsilon$.
A general construction of $\mathcal{S}_k^\varepsilon$ is given in Appendix C \cite{supp}.

With this discretization, the original optimization problem can be replaced by a finite search, allowing for rigorous error control.
In particular, the deviation from the true optimum is bounded as
\begin{equation}
    0 \leq \|\rho - \sigma^\varepsilon\| - \|\rho - \sigma^*\| \leq \varepsilon\,,
\end{equation}
where $\sigma^*$ denotes the optimal solution to \cref{eq:opt}, and $\sigma^\varepsilon$ is the solution obtained from $\mathcal{S}_k^\varepsilon$.

This enables us to construct a rigorous entanglement witness as follows:
\begin{proposition}\label{prop:witness_robust}
    Given an entangled state $\rho$ and a separable state $\sigma$, let the product state $\phi$ in ${\mathcal{S}_k^\varepsilon}$ be the closest to the direction $\Lambda =\sigma - \rho$.
    Then the operator
    \begin{eqnarray}\label{eq:witness_robust}
        \mathcal{W} = \frac{\Lambda - (\beta-\epsilon)\cdot\openone}{||\Lambda||}
    \end{eqnarray}
    defines an entanglement witness for the quantum state $\rho$, where $\beta = \tr(\Lambda\ket{\phi}\bra{\phi})$ and $\epsilon = (1-\eta)||\Lambda||$.
    Here, the factor $\eta$ depends on $\varepsilon$ and the construction of $\mathcal{S}_k^\varepsilon$.
\end{proposition}
\begin{proof}
    See the proof in Appendix D \cite{supp}.
\end{proof}

While the search over the $\varepsilon$-net space can be computationally expensive, it is only required as a subroutine within the CG-based framework and not in every iteration.
Its effectiveness depends on the quality of the reference separable state $\sigma$, which serves as an anchor for constructing the witness.

In practice, this structured search complements the fast heuristic method.
The heuristic CG procedure is used initially to rapidly detect entanglement in most practical scenarios.
If the ratio $r_t$ fails to converge to the detection threshold, or when experimental certification or precise boundary identification is needed, the rigorous entanglement witness can be invoked to ensure validity.

\textit{Applications for closing gap.---}
To demonstrate the power of our method, particularly its capacity for rigorous entanglement certification, we consider the problem of estimating the white noise robustness of entangled states.
This task serves as a standard and widely used benchmark for evaluating entanglement and separability detection tools.
It is to determine the minimum white noise level $p$ required to render the mixed state $\rho(p)=(1-p)\rho+p\mathbb{I}/d$ separable; that is,
$p_{\mathrm{sep}}=\min\{p\in[0,1]:\rho(p)\in S_k\}$.

As previously mentioned, the CG-based framework can also certify separability using a geometric reconstruction procedure~\cite{SG18, GB03}.
In our framework, it can be directly achieved based on the same algorithmic results of \cref{eq:opt} as the following \cref{prop:sep_ball_Id}.
This allows us to construct a unified toolbox for closing the gap between entanglement and separability bounds in the white noise robustness problem.

\begin{proposition}\label{prop:sep_ball_Id}
Given a noisy (separable or entangled) quantum state $\rho(p) = (1 - p)\rho + p\,\openone/d$, and suppose $\|\rho(p) - \sigma\| = \delta$ for some $\sigma \in \mathcal{S}_k$.
Then the state $\rho(\frac{p + \epsilon}{1 + \epsilon})$ must be separable, where $\epsilon \geq \delta/a$, and $a$ is the radius of the separable ball.
\end{proposition}
\begin{proof}
    See the proof in Appendix E \cite{supp}.
\end{proof}

\begin{figure}[t]
  \includegraphics[width=\columnwidth]{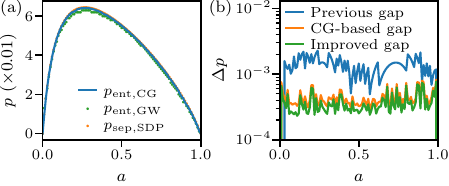}
  \caption{(a) Entanglement bounds $p_{\rm ent, CG}$ (blue line) on the white noise robustness of Horodecki states $\rho^H(a)$, obtained via our CG method, compared with previous entanglement bounds based on a generalization of the Wootters formula $p_{\rm ent, GW}$ (blue dots) and separability bounds based on SDP methods $p_{\rm sep, SDP}$ obtained by the SDP-based method of Ohst et al.~\cite{OYGN24} (orange dots).
  (b) Comparison of the gap $\Delta p$ between entanglement and separability bounds: blue shows the range $[p_{\rm ent, GW}, p_{\rm sep, SDP}]$, green shows $[p_{\rm ent, CG}, p_{\rm sep, SDP}]$, and orange shows $[p_{\rm ent, CG}, p_{\rm sep, CG}]$.}
  \label{fig:horodecki}
\end{figure}

As a concrete example, we apply our method to the Horodecki states $\rho^H(a)$, a class of $3 \times 3$ bound entangled states introduced by P.~Horodecki~\cite{Horo97}.
These states are entangled for all $0 < a < 1$, yet cannot be detected by the PPT criterion.
The explicit form of $\rho^H(a)$ is provided in Appendix F1 \cite{supp}.

As shown in \cref{fig:horodecki}(a), our CG-based approach yields significantly improved entanglement bounds $p_{\rm ent, CG}$, outperforming the best known bounds from the generalized Wootters formula~\cite{Chen_2012_estimating}.
Notably, the new entanglement thresholds approach the previously established separability limits $p_{\rm sep, SDP}$ obtained by Ohst et al.~\cite{OYGN24}, nearly closing the gap.
\cref{fig:horodecki}(b) further quantifies the improvement.
Our method substantially reduces the entanglement-separability gap.
Moreover, the entanglement and separability thresholds obtained from our CG-based approach yield a gap that is comparable to the tightest known gap between entanglement and separability bounds, defined by our results and SDP methods.
We also note that in the parameter regime $0 < a < 0.02$, previous results exhibited inconsistencies between the entanglement and separability bounds.
Our method not only narrows the gap but also resolves this contradiction, improving both numerical precision and physical consistency.

\textit{Applications under arbitrary noise channels.---}
From an algorithmic perspective, estimating the white noise robustness of a quantum state $\rho$ is computationally equivalent to certifying entanglement of its noisy counterpart $\rho(p)$.
In SDP-based approaches, this typically requires only an additional variable in the optimization.
In our CG-based framework, the parameter $p$ can be updated dynamically during iterations based on the ratio $r$.

\begin{figure}[t]
  \includegraphics[width=0.70\columnwidth]{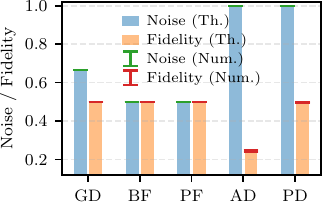}
  \caption{Entanglement robustness thresholds for the Bell state $(\ket{00}+\ket{11})/\sqrt{2}$ under various quantum noise channels: global depolarizing (white noise, GD), bit flip (BF), phase flip (PF), amplitude damping (AD), and phase damping (PD).
  Except for white noise (applied to both qubits), all channels act on one side of the state.
  Bars indicate theoretical thresholds for the noise strength (blue) and the corresponding fidelity with the ideal Bell state (orange).
  Error bars represent the certified interval between the entanglement and separability bounds obtained by our numerical algorithm.
  }
  \label{fig:channel}
\end{figure}

Both approaches can generalize to arbitrary linear noise models of the form
$(1 - p)\rho + p \sigma\,, \sigma \in \mathcal{H}$.
However, in realistic settings, quantum noise is often nonlinear and must be modeled as a quantum channel $\Phi$, defined via Kraus operators as $\Phi(\rho) = \sum_i K_i \rho K_i^\dagger$ with $\sum_i K_i^\dagger K_i \leq \openone$.
For instance, the bit-flip channel is defined by $K_0 = \sqrt{1-p}\,\openone$ and $K_1 = \sqrt{p}\,X$.

Different noise models yield different entanglement thresholds, which are of practical importance for experimental applications, e.g., determining fidelity thresholds for maintaining functionality in quantum protocols.
The robustness of entanglement under general quantum channels lies beyond the reach of current SDP methods, but can be naturally addressed within our CG-based framework using the same convergence criterion based on $r_t$.

\cref{fig:channel} presents the entanglement robustness thresholds for the Bell state $(\ket{00} + \ket{11})/\sqrt{2}$ under several common quantum noise channels relevant to quantum computation and teleportation.
In each case, our numerical method yields certified entanglement and separability bounds with tight intervals, represented by error bars.
The results show excellent agreement with known theoretical thresholds, demonstrating the accuracy and reliability of our method.
Further details and comparisons with analytical values are provided in Appendix F2 \cite{supp}.

\textit{Conclusion.---}
We have developed a unified framework for multipartite entanglement characterization based on the conditional gradient method, encompassing both heuristic and rigorous procedures.
Our approach enables rapid entanglement detection for systems of up to ten qubits and constructs rigorous entanglement witnesses with explicit numerical error control.
By applying this framework to the problem of white noise robustness for bound entangled states, we obtained the optimal entanglement thresholds, thereby closing the gap between entanglement and separability thresholds.
We have further demonstrated the applicability of our method in analyzing entanglement robustness under general quantum noise channels, providing accurate thresholds in regimes beyond the reach of SDP-based techniques.
These results establish CG methods as a versatile and scalable tool for entanglement analysis in current experimental scenarios.

Beyond these results, our framework suggests several promising directions for future work.
Exploiting symmetries in quantum states and measurements may simplify optimization and improve interpretability, especially for highly symmetric states such as GHZ or Dicke states.
The method can be extended to compute robustness measures by optimizing over both the separable space and the noisy space.
Incorporating matrix product state (MPS) representations may scale the method to hundreds of qubits.
Additionally, splitting methods can be integrated to address incomplete or partial information, broadening the applicability of our framework to a wider range of quantum characterization settings.

\acknowledgments
We are grateful to Márton Naszodi, H.~Chau Nguyen, Ties-A.~Ohst, Jiangwei Shang, Tamás Vértesi, and Liding Xu for helpful discussions.
This work was supported by the DFG Cluster of Excellence MATH+ (EXC-2046/1, Project No.~390685689) funded by the Deutsche Forschungsgemeinschaft (DFG), and partly funded within the QuantERA II Programme that has received funding from the European Union’s Horizon 2020 research and innovation programme under Grant Agreement No 101017733 (VERIqTAS).

%

\onecolumngrid
\newpage
\appendix

\section{Frank-Wolfe method}\label{app:fw}

Conditional Gradient (CG) methods, also known as Frank-Wolfe (FW) methods~\citep{FW56, LP66}, are a class of projection-free first-order algorithms for constrained convex optimization problems of the form
$$
    \min_{x \in \mathcal{X}} f(x)\,,
$$
where $f$ is a smooth convex function and $\mathcal{X}$ is a compact convex set.
Instead of performing projections onto $\mathcal{X}$, which can be computationally expensive, CG methods access the feasible region via a Linear Minimization Oracle (LMO).
At each iteration, the algorithm solves a linear subproblem to compute the so-called FW vertex $$v_t = \argmin_{v \in \mathcal{X}} \langle \nabla f(x_t), v \rangle\,.$$
The next iterate is then a convex combination of $v_t$ and $x_t$, which ensures feasibility.
While FW methods are attractive for their simplicity and projection-free structure, their vanilla form requires an LMO call at every iteration and achieves only a sublinear convergence rate of $\mathcal{O}(1/t)$ in the general convex case.

To improve convergence, active-set variants of CG methods maintain an explicit convex combination of the current iterate over a set of previously used vertices, the so-called active set $\mathcal{A}_t$.
Notable examples include Away-step Frank-Wolfe~\citep{guelat1986some}, Fully-Corrective Frank-Wolfe~\citep{holloway1974extension}, and Blended Pairwise Conditional Gradients (BPCG)~\citep{tsuji2022pairwise}.
Additionally to the standard FW step, these methods perform local updates changing the weights of the vertices.
This enables dropping previously selected vertices from the active set, leading to sparser solutions and faster convergence.
In the case of polytopes and strongly convex objectives, these methods can achieve linear convergence.

BPCG is a particularly efficient instance of the Frank-Wolfe method.
In each iteration, it computes the local FW vertex $s_t$ and the away vertex $a_t$,
\begin{align*}
    s_t &= \argmin_{v \in \mathcal{A}_t} \langle \nabla f(x_t), v\rangle\,,\\
    a_t &= \argmax_{v \in \mathcal{A}_t} \langle \nabla f(x_t), v\rangle\,.
\end{align*}
Unlike the global FW vertex, computing these vertices does not require a call to the LMO but can be done efficiently by enumerating the active set.

BPCG then compares the standard FW direction $d_t^{FW} = v_t - x_t$ with the local pairwise direction $d_t^{LPW} = s_t - a_t$ based on their respective gap value.
The Frank-Wolfe gap $g_t^{FW}=\langle \nabla f(x_t),x_t-v_t \rangle$ and the local pairwise gap $g_t^{LPW} = \langle \nabla f(x_t),a_t-s_t\rangle$ measure the alignment of each direction with the negative gradient and thus their potential to decrease the objective.
BPCG proceeds in the direction corresponding to the larger gap.

Furthermore, BPCG uses a lazification technique inspired by~\citet{pok17lazy} storing an estimate of the Frank-Wolfe gap $\Phi_t$.
By this, one avoids LMO calls in iterations where the local update is used.
This is especially useful in the context of complicated feasible regions, where the LMO is expensive, which is precisely our case in this work.

To accelerate our method even further, we equip BPCG with quadratic correction steps (QC) recently proposed in~\citet{halbey2025efficientquadraticcorrectionsfrankwolfe}.
In cooperation with our group they conducted preliminary experiments on our problem and included the results in their work.
The QC step is tailored for quadratic objectives, $f(x) = \frac12 x^T Ax + b^T x$, like the considered projection problem.
While the local pairwise step updates only two weights in each iteration, QC aims to optimize the weights of the whole active set in a single step.
Minimizing a given objective over the convex hull of the current active set is in general as hard as the original problem.
However, minimizing a quadratic objective over the affine hull of given vertices can be done by solving a linear system,
$$
    \sum_{u \in \mathcal{A}_t} \lambda[u] \langle  A u, v-w \rangle = 0 \quad \forall v \in \mathcal{A}_t \setminus \{w\}, \quad \sum_{u \in \mathcal{A}_t} \lambda[u] = 1\,,
$$
where $\lambda[u]$ denotes the barycentric coordinate of the vertex $u$ and $w$ is a fixed vertex in $\mathcal{A}_t$.
The affine minimizer only lies in the convex hull of the active set (and is therefore feasible) if and only if the barycentric coordinates $\lambda[u]$ are all non-negative.

To handle this issue we follow the QC-MNP approach inspired by the Minimum Norm Point (MNP) algorithm by~\citet{wolfe1976normpoint}.
One considers the line segment between the current iterate and the affine minimizer.
Using a simple ration test, one can find the point on the line segment and in the convex hull that is closest to the affine minimizer.
In practice, we perform a QC-step every time a new vertex is added to the active set, otherwise we continue with cheap pairwise updates in BPCG.

The framework of the algorithm we actually used for entanglement characterization is as follows:
\begin{algorithm}[H]
    \renewcommand\thealgorithm{}
    \floatname{algorithm}{}
    \caption{Lazified BPCG with QC-MNP}
    \label{alg:correctiveFW}
    \begin{algorithmic}[1]
        \Require Distance function $f(x)$; initial point $\x_0 \in \mathcal{S}_k$
        \State $\Phi_0 \gets \max_{v \in \mathcal{X}} \langle \nabla f(x_0), x_0 - v\rangle/2$
        \State $\mathcal{A}_0 \gets \{x_0\}$
        \State $b=$ false
        \For{$t = 0, \ldots, T$}
            \State $a_t \gets \argmax_{v\in \mathcal{A}_t} \langle \nabla f(x_t), v\rangle$ \Comment{away vertex}
            \State $s_t \gets \argmin_{v\in \mathcal{A}_t} \langle \nabla f(x_t), v\rangle$ \Comment{local FW vertex}
            \If{$\langle \nabla f(x_t), a_t - s_t\rangle \geq \Phi_t$}
                \If{$b$}
                    \State $x_{t+1}, \mathcal{A}_{t+1} \gets \text{QC-MNP}(\mathcal{A}_t)$
                    \State $\Phi_{t+1} \gets \Phi_t$
                    \State $b \gets \text{false}$
                \Else \Comment{local pairwise update}
                    \State $d_t \gets a_t - s_t$
                    \State $\gamma_t^* \gets \lambda[a_t]$
                    \State $\gamma_t \gets \argmin_{\gamma \in [0, \gamma_t^*]} f(x_t - \gamma d_t)$
                    \State $x_{t+1} \gets x_t - \gamma_td_t$
                    \State $\Phi_{t+1} \gets \Phi_t$
                    \If{$\gamma_t < \gamma_t^*$}
                        \State $\mathcal{A}_{t+1} \gets \mathcal{A}_{t}$ \Comment{descent step}
                    \Else
                        \State $\mathcal{A}_{t+1} \gets \mathcal{A}_{t} \setminus \{a_t\}$ \Comment{drop step}
                    \EndIf
                \EndIf
            \Else
                \State $v_t \gets \argmin_{v\in \mathcal{X}} \langle \nabla f(x_t), v\rangle$ \Comment{global FW vertex}
                \State $b \gets (v_t \in \mathcal{A}_t)$
                \If{$\langle \nabla f(x_t), x_t -v_t\rangle \geq \Phi_t/2$}
                    \State $d_t \gets x_t - v_t$
                    \State $\gamma_t \gets \argmin_{\gamma \in [0, 1]} f(x_t - \gamma d_t)$
                    \State $x_{t+1} \gets x_t - \gamma_td_t$
                    \State $\Phi_{t+1} \gets \Phi_t$
                    \State $\mathcal{A}_{t+1} \gets \mathcal{A}_{t} \cup \{v_t\}$ \Comment{FW step}
                \Else
                    \State $x_{t+1} \gets x_t$
                    \State $\Phi_{t+1} \gets \Phi_t/2$
                    \State $\mathcal{A}_{t+1} \gets \mathcal{A}_{t}$ \Comment{gap step}
                \EndIf
            \EndIf
        \EndFor
        \Ensure $\x_1, \ldots, \x_T \in \mathcal{X}$
    \end{algorithmic}
\end{algorithm}

\section{Examples for detecting different entanglement structures}\label{app:gme}
Here, we recall the definition of multipartite entanglement.
A multipartite pure state is called $k$-separable if it can be written as $\ket{\psi^{k\rm s}} = \ket{\phi_1} \otimes \ket{\phi_2} \otimes \cdots \otimes \ket{\phi_k}$, where each $\ket{\phi_i}$ may itself be entangled within a subset of the total system.
A mixed state is $k$-separable if it can be expressed as a convex combination of $k$-separable pure states: $\sigma^{k\rm s} = \sum_i \lambda_i \ket{\psi^{k\rm s}_i}\bra{\psi^{k\rm s}_i}$, where the coefficients $\lambda_i$ form a probability distribution.
Usually, if $n$-partite quantum states are not $n$-separable, then they are called entangled states.
And non-$2$-separable states are called GME.

In the main text, we focus on the detection of multipartite entanglement of $n$-qubit states against full separability.
In this section, we illustrate that CG methods can naturally handle rich structures of multipartite entanglement.
We consider the white noise robustness problem for $\mathcal{S}_k$ with $k=2$ and $3$ for several common quantum states.
Note, when $k=2$, it's the case of genuine multipartite entanglement (GME), where detection of GME is an important problem in the field.

\begin{table}[htb]
    \centering
    \renewcommand{\arraystretch}{1.1} %
    \setlength{\tabcolsep}{10pt} %
    \begin{tabular}{|>{\raggedright\arraybackslash}p{3.5cm}
                    |l|l|}
        \hline
        \multirow{2}{*}{Quantum states} & \multicolumn{2}{c|}{Entanglement bound $\pm$ gap} \\
        \cline{2-3}
        &   \qquad\qquad$k=2$ (GME)&    \qquad\qquad$k=3$\\
        \hline
        GHZ(5 Qubits) &
        0.508 + 0.016 ($\frac{16}{31}\approx 0.516$~\cite{Guhne_2010_separability}) &
        0.760 + 0.006 ($\frac{16}{21}\approx 0.762$~\cite{Chen_2018_necessary}) \\
        GHZ(4 Qubits) &
        0.528 + 0.005 ($\frac{8}{15}\approx 0.533$~\cite{Guhne_2010_separability}) &
        0.797 + 0.005 ($\frac45=0.8$~\cite{Chen_2018_necessary}) \\
        W(5 Qubits) &
        0.515 + 0.070 (0.478~\cite{Gao_2010_detection}, -) &
        0.739 + 0.013 (-, -) \\
        W(4 Qubits) &
        0.528 + 0.003 (0.526~\cite{Jungnitsch_2011_taming}, 0.545~\cite{SG18}) &
        0.722 + 0.033 ($\frac{216-16\sqrt{6}}{235}\approx 0.752$~\cite{Ge_2021_tripartite}) \\
        Dicke(5 Qubits, 2 ex.) &
        0.518 + 0.047 (0.516~\cite{Ananth_2016_nonkseparability}, -) &
        0.717 + 0.055 (0.615, -) \\
        Dicke(4 Qubits, 2 ex.) &
        0.540 + 0.003 (0.539~\cite{Jungnitsch_2011_taming}, -) &
        0.769 + 0.032 (0.788, -) \\
        \hline
    \end{tabular}
    \caption{
        Entanglement robustness thresholds for several multipartite quantum states against white noise, with respect to $k$-separability for $k=2$ (GME) and $k=3$.
        Each entry shows the entanglement bound $p_{\rm ent}$ and the gap $\Delta p$ to separable bound, obtained by our CG-based framework.
        Known entanglement and separable thresholds from the literature are given in parentheses for comparison.
        Numbers in parentheses represent the exact thresholds.
    }
    \label{tab:result}
\end{table}

\section{Construction on an \texorpdfstring{$\varepsilon$}{epsilon}-net of the set of pure separable states}\label{app:net}

With~\cite{AS17} and in particular Lemma 9.2 therein, the construction of an $\varepsilon$-net of the set of pure separable states can be reduced to the construction of an $\varepsilon$-net of a corresponding hypersphere.
In this appendix we then concentrate on this case, giving an explicit construction, in contrast with the arguments used in~\cite{AS17}.

\subsection{Definitions and notations}

Let $H$ be a $d$-dimensional normed space with unit ball $B$ and unit sphere $S$.
Let $K$ be a nonempty centrosymmetric compact convex set in $H$ and $X=\{x_i\}_{i=1}^m$ be a subset of $\ext(K)$ such that $0\in\conv(X)=P$.

\begin{definition}\label{def:shrink}
  The homothetic distance, also called the geometric distance or the shrinking factor, of $X$ with respect to $K$ is the maximum $\eta$ such that $\eta K\subset P$.
\end{definition}

\begin{definition}
  $X$ is an inner $\epsilon$-approximation of $K$ if there exists $\epsilon>0$ such that $\frac{1}{1+\epsilon}K\subset P$.
\end{definition}

\begin{definition}
  $X$ is an $\varepsilon$-net of $\ext(K)$ if there exists $\varepsilon>0$ such that $\ext(K)\subset X+\varepsilon B$.
\end{definition}

In this appendix we focus on the case where $H$ is real and $K=B$ is the unit ball.

\begin{observation}
  An inner $\epsilon$-approximation of $B$ has homothetic distance $\eta\geq\frac{1}{1+\epsilon}$ with respect to $B$.
\end{observation}

\begin{observation}
  An $\varepsilon$-net of $S$ has homothetic distance $\eta\geq\sqrt{1-\frac{\varepsilon^2}{2}}$ with respect to $B$.
\end{observation}

\subsection{Implicit construction}

In~\cite[Lemma~4.10]{Pis89}, the following result is stated, whose proof we recall for completeness.

\begin{proposition}
  There exists an $\varepsilon$-net of $S$ with cardinality smaller than $(1+\frac{2}{\varepsilon})^d$.
  \label{prop:implicit}
\end{proposition}

\begin{proof}
  Let $X$ be a maximal subset of $S$ such that ${\|x_i-x_j\|\geq\varepsilon}$ for all $i\neq j$.
  Clearly, by maximality, $X$ is an $\varepsilon$-net of $S$.
  To derive an upper bound on its cardinality $m$, we note that the balls $x_i+\frac{\varepsilon}{2}B$ are disjoint and included into $(1+\frac{\varepsilon}{2})B$.
  Therefore,
  \begin{equation}
    \sum_{i=1}^m\vol\left(x_i+\frac{\varepsilon}{2}B\right)\leq\vol\left[\left(1+\frac{\varepsilon}{2}\right)B\right]\,,
    \quad\text{hence}\quad
    m\left(\frac{\varepsilon}{2}\right)^d\vol(B)\leq\left(1+\frac{\varepsilon}{2}\right)^d\vol(B)\,,
  \end{equation}
  and this implies $m\leq(1+\frac{2}{\varepsilon})^d$ as announced.
\end{proof}

To the best of our knowledge, no better scaling is currently known, but the construction is implicit, so that its usage in practice is somehow limited.
In the next subsection, we present an alternative construction that is fully explicit, and whose scaling is even better.

\subsection{Explicit construction}

We first recall without proof the main result from~\cite{EG00} on the edgewise subdivision of a simplex.

\begin{proposition}
  For every integer $n\geq1$ and every $d$-simplex $\sigma$ there is a subdivision of $\sigma$ into $n^d$ $d$-simplices of equal volume.
  \label{prop:subdivision}
\end{proposition}

\begin{theorem}
  Let $X\subset S$ be an inner $\epsilon$-approximation of $B$ and consider a simplicial complex triangulating the facets of $P=\conv(X)$ and such that the circumcenter of each simplex lies inside this simplex.
  For an integer $n\geq2$, construct the set $X_n$ by applying \cref{prop:subdivision} to all of these simplices and projecting the resulting vertices onto $S$.
  Then $X_n$ is an inner $\epsilon_n$-approximation of $B$ for $\epsilon_n$ such that
  \begin{equation}
    (1+\epsilon_n)^2=1+\frac{d-1}{2d}D_n^2(1+\epsilon)^2\,,
  \end{equation}
  where $D_n$ is the maximal diameter of the subsimplices before projection onto $S$.
  \label{thm:main}
\end{theorem}

\begin{proof}
  Let $\Delta$ be one of the $(d-1)$-simplices obtained after subdivision of the facets, but before projection onto $S$.
  Denote $c$ the center of its circumscribed sphere (which by assumption is also the center of its minimum enclosing ball) and $a$ one point on this sphere.
  By Jung's theorem, we have that
  \begin{equation}
    \|a-c\|\leq\frac{d-1}{2d}\diam(\Delta)^2\leq\frac{d-1}{2d}D_n^2\,.
  \end{equation}
  Since $c$ lies on a facet of $P$, we also have $\|c\|\geq\frac{1}{1+\epsilon}$.
  Combining these two inequalities and by the Pythagorean theorem (see \cref{fig:proof}), we obtain
  \begin{equation}
    \left\|\frac{c}{\|a\|}\right\|^2=\frac{\|c\|^2}{\|c\|^2+\|a-c\|^2}\geq\frac{\left(\frac{1}{1+\epsilon}\right)^2}{\left(\frac{1}{1+\epsilon}\right)^2+\frac{d-1}{2d}D_n^2}\,.
  \end{equation}
  Since the closest point to the origin of the projection of $\Delta$ onto $S$ is $\frac{c}{\|a\|}$, this bound is valid for $X_n$.
\end{proof}

\begin{figure}[h]
  \centering
  \includegraphics{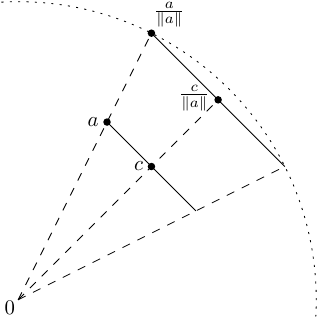}
  \caption{Setup for the proof of \cref{thm:main}.}
  \label{fig:proof}
\end{figure}

\begin{proposition}
  Consider a regular $(d-1)$-simplex with edge length $s$.
  For $d\leq3$ or $n=1$, the diameters of its subsimplices are all $\frac{s}{n}$; in the other cases, they are at most $\frac{s\sqrt{2}}{n}$.
  \label{prop:regular}
\end{proposition}

\begin{proof}
  The diameter of a simplex is simply the largest distance between two of its vertices, which immediately gives $D_1=s$.
  From the proof of the independence lemma in~\cite{EG00}, we know that the vectors defining the subsimplices are, up to a factor $n$, the same as those defining the initial simplex, the various shapes coming from a permutation of the order of these vectors.
  For $d\leq3$, this permutation does not affect the edge length of the different shapes.
  For higher dimensions, the worst case gives a factor $\sqrt2$ by the Pythagorean theorem.
\end{proof}

\begin{example}
  Let $X$ be the set of vertices of the $d$-dimensional cross-polytope.
  Observe first that its homothetic distance with respect to the unit ball is $\eta=\frac{1}{\sqrt{d}}$, corresponding to $\epsilon=\sqrt{d}-1$.
  Then choose an integer $n\geq2$ and define $X_n$ according to the procedure  in \cref{thm:main}.
  All $2^d$ facets of $P$ are regular $(d-1)$-simplices with edge length $\sqrt2$.
  Each facet gives rise to $\binom{n+d}{d}$ vertices of $X_n$ (the $(n+1)$-th $(d-1)$-simplex number), so that the total cardinality of $X_n$ is lower than $2^d\binom{n+d}{d}$.
  With \cref{prop:regular,thm:main}, we obtain that the homothetic distance $\eta_n$ of $X_n$ with respect to the unit ball is such that
  \begin{equation}
    \eta_n\geq\frac{n}{\sqrt{n^2+2(d-1)}}\,,
  \end{equation}
  so that $X_n$ is an $\varepsilon$-net of $S$ for
  \begin{equation}
    \varepsilon=\sqrt{\frac{2}{1+\frac{n^2}{2(d-1)}}}\sim\frac{2\sqrt{d-1}}{n}\,.
  \end{equation}
  This gives the following scaling for the cardinality of this family of $\varepsilon$-nets:
  \begin{equation}
    \frac{2}{(d-1)!}\left(\frac{4\sqrt{d-1}}{\varepsilon}\right)^{d-1}\,,
  \end{equation}
  which improves on \cref{prop:implicit} as announced.
\end{example}

\section{Proof of \texorpdfstring{\cref{prop:witness_robust}}{Proposition 3}}\label{App:proof_robust}
\cref{prop:witness_robust} is built on the construction on $\varepsilon$-net of quantum space, which is discussed above.
As the product state $\phi$ is the closest state to the direction $\Lambda$ in $\mathcal{S}_k^\varepsilon$, there is
\begin{equation}
    \tr(\Lambda\ket{\phi_i}\bra{\phi_i}) \leq \tr(\Lambda\ket{\phi}\bra{\phi}) = \beta\,,\quad \forall\phi_i \in {\rm conv}(\mathcal{S}_k^\varepsilon)\,.
\end{equation}
According to \cref{def:shrink}, we have
\begin{equation}
    \eta \mathcal{S}_k \subset {\rm conv}(\mathcal{S}_k^\varepsilon) \subset \mathcal{S}_k\,.
\end{equation}
Thus, for every $\tau \in \mathcal{S}_k$, we have
\begin{equation}
    \begin{aligned}
    \tr(\Lambda\tau)
        &=      \tr(\Lambda\, \eta \tau) + (1-\eta)\tr(\Lambda\tau)   \\
        &\leq   \tr(\Lambda\, \ket{\phi^*}\bra{\phi^*}) + (1-\eta)\tr(\Lambda\tau)   \\
        &\leq   \tr(\Lambda\, \ket{\phi^*}\bra{\phi^*}) - (1-\eta)||\Lambda||\,.
    \end{aligned}
\end{equation}
Note the final inequality applying the Cauchy-Schwarz inequality and $\tr(\tau)=1$.
Thus we obtain the witness $\mathcal{W}$ in \cref{eq:witness_robust} such that
\begin{eqnarray}
    \tr(\mathcal{W}\tau) \geq 0 \,, \quad \forall \tau \in \mathcal{S}_k\,.
\end{eqnarray}

\section{Proof of \texorpdfstring{\cref{prop:sep_ball_Id}}{Proposition 4}}\label{app:sep_ball}
\begin{proof}
In Ref.~\cite{SG18}, a geometric reconstruction approach was developed within the CG-based framework for separability certification.
In short, based on a quantum state $\rho$, we can construct a state $\rho_t = (1-\epsilon)\rho - \epsilon\openone/d$.
For the constructed state $\rho_t$, the CG method can give us the separable state $\sigma_t$ close to $\rho_t$, as a numerical result.
Extrapolating $\rho$ against $\rho_t$ to achieve $\rho_x$ such that the direction of $\rho_t - \sigma_t$ and $\rho_x -\openone/d$ are parallel.
Then if the distance $r_t = \|\rho_t - \sigma_t\|$ is smaller than the radius $a$ of the separate ball~\cite{GB03}, the state $\rho_x$ is separable, therefore the state $\rho$ is separable as well.
An illustration of the geometric reconstruction is shown in \cref{fig:sep_ball}.

\begin{figure}[tb]
  \includegraphics[width=0.40\columnwidth]{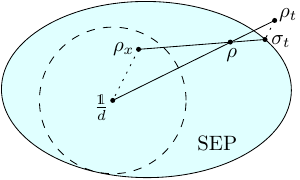}
  \caption{Geometric reconstruction for separability certification within the CG-based framework.}
  \label{fig:sep_ball}
\end{figure}

For the problem of estimating the white noise robustness of an entangled state $\ket{\psi}$, let the constructed state be the noisy state we consider, i.e.,
\begin{eqnarray}
    \rho_t = \rho(p) = (1-p)\ketbra{\psi} + p \openone/d^2\,.
\end{eqnarray}
Then as long as
\begin{eqnarray}
    r_t = \frac{1}{\epsilon}\|\rho_t-\sigma_t\| \leq a\,,
\end{eqnarray}
the state
\begin{eqnarray}
    \rho = \frac{1}{1+\epsilon}\rho_t+\frac{\epsilon}{1+\epsilon}\frac{\openone}{d}
    =\left(1-\frac{\epsilon+p}{1+\epsilon}\right)\ketbra{\psi}+\frac{\epsilon+p}{1+\varepsilon}\frac{\openone}{d}
    =\rho\left(\frac{p+\epsilon}{1+\epsilon}\right)
\end{eqnarray}
will be separable.
\end{proof}

\section{Details of applications}
As shown in \cref{fig:shortcut} in the main text, our method, based on the dual gap efficiency criterion ($r_t$, Corollary~\ref{coro:criterion}) rather than relying on full convergence, significantly accelerates entanglement detection. This capability is essential for scaling to large quantum systems in practical scenarios. A careful reader will notice the different detection speeds for the GHZ, $W$, and Dicke states with identical noise levels, originating from the interplay between each state’s structure and the algorithm’s complexity.

We first analyze the performance on pure GHZ, $W$, and Dicke states. 
The total time consumption of a CG method is primarily governed by the convergence rate (number of iterations $t$) and the per-iteration cost (dominated by the LMO solving). 
As shown in \cref{fig:converge}, the convergence rate (reduction of the objective function) for the pure GHZ states is typically faster than for $W$ and Dicke states. 
However, \cref{fig:time} shows that pure GHZ states consume significantly more time per iteration. 
This higher cost is reflected by the steeper slope of the time curve, and is also accompanied by a significantly larger active set size compared to $W$ and Dicke states. 
This high per-iteration cost is ultimately due to the complexity of solving the LMO for the GHZ state's structure, a factor that will be more clearly reflected in the following noisy cases. 
Despite this higher per-iteration cost, relying on the dual gap efficiency criterion $r_t$ (as shown in \cref{fig:shortcut_app}(a)), pure GHZ states can be typically detected within only two iterations, fewer than the five iterations required for the $W$ and Dicke states, resulting in a faster total detection time in the pure state regime.

\begin{figure}[tb]
  \centering
  \includegraphics[width=0.90\columnwidth]{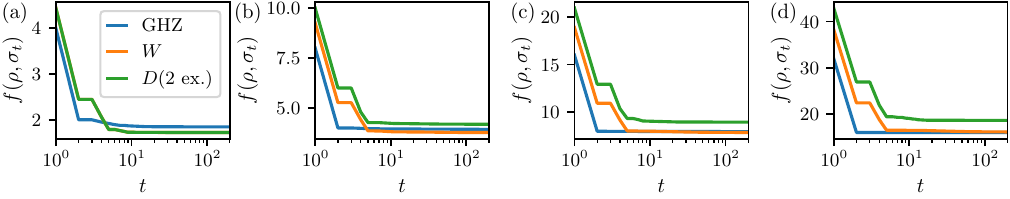}
  \caption{The convergence behavior of the objective function $f(\rho,\sigma_t)$ for pure GHZ states, $W$ states, and Dicke states (with two excitations) as a function of the number of iterations $t$. Results are shown for qubit numbers (a) $N=3$, (b) $N=4$, (c) $N=5$, and (d) $N=6$.
  }
  \label{fig:converge}
\end{figure}
\begin{figure}[tb]
  \centering
  \includegraphics[width=0.99\columnwidth]{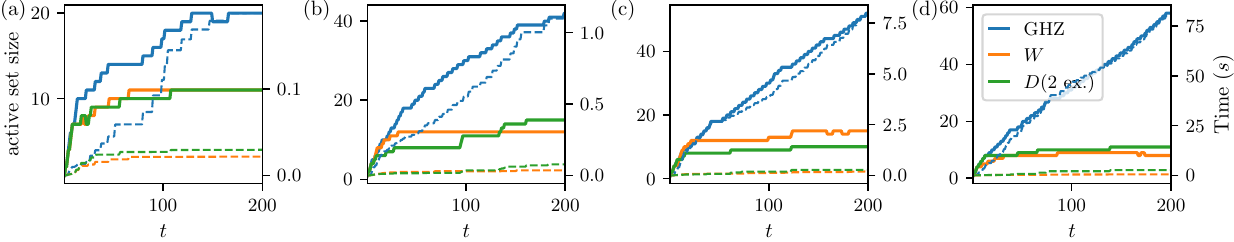}
  \caption{The size of the active set (solid lines) and the cumulative time consumption (dashed lines) for pure GHZ, $W$, and Dicke states. Results are shown for qubit numbers (a) $N=3$, (b) $N=4$, (c) $N=5$, and (d) $N=6$.}
  \label{fig:time}
\end{figure}
\begin{figure}[tb]
  \centering
  \includegraphics[width=0.70\columnwidth]{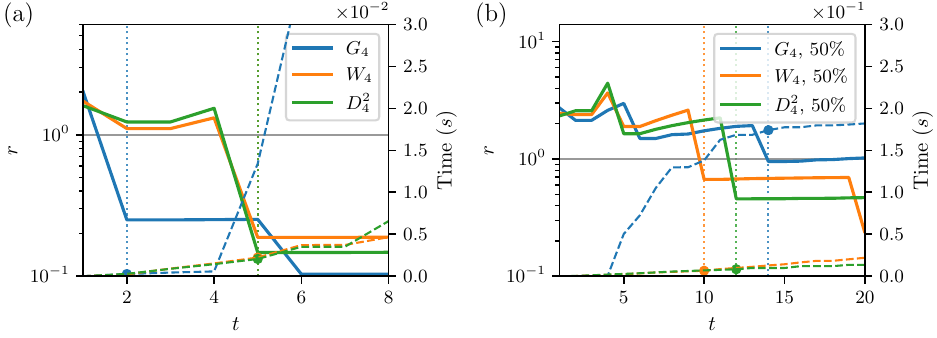}
  \caption{The rapid detection performance for the 4-qubit GHZ state, $W$ state, and Dicke state (with two excitations). The figure plots the dual gap efficiency $r$ (solid line) and the cumulative time (dashed line) for (a) pure states and (b) states with $50\%$ white noise. The entanglement is detected when $r < 1$ (gray line), and the corresponding rapid detection time is marked (dot).}
  \label{fig:shortcut_app}
\end{figure}

The introduction of white noise significantly alters the computational landscape. 
As illustrated in \cref{fig:converge_noise}, taking the $50\%$ white noise as an example, the noisy GHZ states initially converge quickly but then exhibit a slowing trend before final convergence, indicating increased difficulty near the minimum. 
While the size of the active set is similar across all state types, the time required per iteration (primarily the LMO cost) is significantly different.
\cref{fig:time_noise} shows that GHZ states consume more total time than $W$ and Dicke states, especially in larger systems.
This performance gap arises because the LMO cost strongly depends on the state structure: compared with $W$ and Dicke states, GHZ states require more repetitions to converge within the alternating-minimization scheme of the LMO, leading to slower heuristic optimization.
And this structure difference can be reflected by the different sparsity of correlation tensors of these states.
Hence, as observed in \cref{fig:shortcut_app}(b), the high LMO cost per step causes the entanglement detection of GHZ states to be generally slower than for $W$ and Dicke states at the same noise level.

\begin{figure}[tb]
  \centering
  \includegraphics[width=0.90\columnwidth]{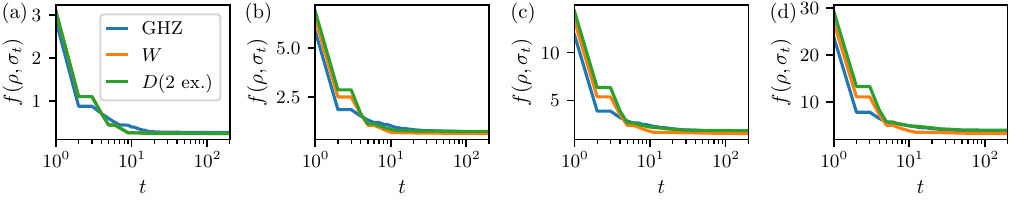}
  \caption{The convergence behavior of the objective function $f(\rho,\sigma_t)$ for noisy GHZ states, $W$ states, and Dicke states (with two excitations) as a function of the number of iterations $t$ ($50\%$ white noise). Results are shown for qubit numbers (a) $N=3$, (b) $N=4$, (c) $N=5$, and (d) $N=6$.}
  \label{fig:converge_noise}
\end{figure}

\begin{figure}[tb]
  \centering
  \includegraphics[width=0.99\columnwidth]{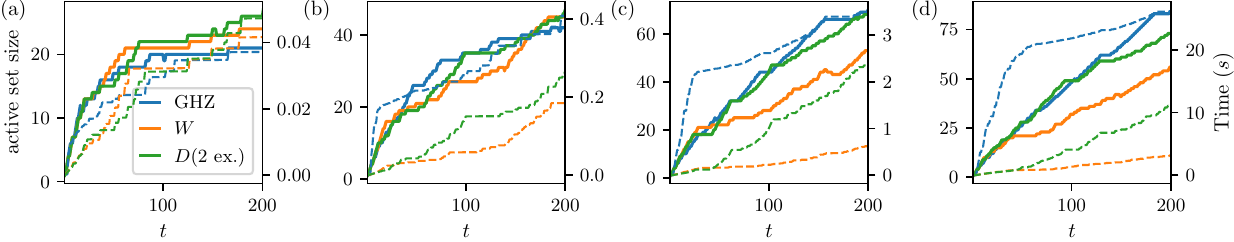}
  \caption{The size of the active set (solid lines) and the cumulative time consumption (dashed lines) for noisy GHZ, $W$, and Dicke states ($50\%$ white noise). Results are shown for qubit numbers (a) $N=3$, (b) $N=4$, (c) $N=5$, and (d) $N=6$.}
  \label{fig:time_noise}
\end{figure}

\subsection{Horodecki state}\label{app:horodecki}
The explicit density matrix of the family of Horodecki states~\cite{Horo97} is
\begin{equation}\label{eq:rhoPH}
\rho^H(a)=\frac1{8a+1}
\left(
    \begin{array}{ccccccccc}
        a     & \cdot & \cdot & \cdot & a     & \cdot & \cdot                  & \cdot & a                      \\
        \cdot & a     & \cdot & \cdot & \cdot & \cdot & \cdot                  & \cdot & \cdot                  \\
        \cdot & \cdot & a     & \cdot & \cdot & \cdot & \cdot                  & \cdot & \cdot                  \\
        \cdot & \cdot & \cdot & a     & \cdot & \cdot & \cdot                  & \cdot & \cdot                  \\
        a     & \cdot & \cdot & \cdot & a     & \cdot & \cdot                  & \cdot & a                      \\
        \cdot & \cdot & \cdot & \cdot & \cdot & a     & \cdot                  & \cdot & \cdot                  \\
        \cdot & \cdot & \cdot & \cdot & \cdot & \cdot & \frac{1+a}{2}          & \cdot & \frac{\sqrt{1-a^2}}{2} \\
        \cdot & \cdot & \cdot & \cdot & \cdot & \cdot & \cdot                  & a     & \cdot                  \\
        a     & \cdot & \cdot & \cdot & a     & \cdot & \frac{\sqrt{1-a^2}}{2} & \cdot & \frac{1+a}{2}
    \end{array}
\right)\!,
\end{equation}
where the dots indicate zeroes. 

We analyze the noisy mixture: $\rho^H(a,p) = (1-p)\rho^H(a) + p \openone/9$.
In the main text, we demonstrated the improved precision of our certification method, leading to a significantly improved gap for the white noise mixture threshold. 
Beyond accuracy, here we analyze the computational efficiency for this challenging class of states. 
The certification time is known to be inversely correlated with the degree of entanglement (positively correlated with the noise level), as states closer to the separable boundary require more time consumption. 
\cref{fig:horodecki_time} provides examples of the time required to detect entanglement of the noisy Horodecki states $\rho^H(a,p)$. 
The results confirm that the typical time consumption for certifying entanglement of these states is of the order of minutes, demonstrating the practical speed and reliability of our framework even for bound entangled states.

\begin{figure}[tb]
  \centering
  \includegraphics[width=0.40\columnwidth]{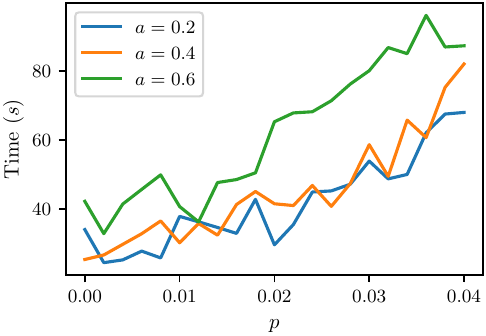}
  \caption{Examples of the computational time required to detect entanglement of the noisy Horodecki states $\rho^H(a,p)$, detected using the dual gap efficiency criterion.}
  \label{fig:horodecki_time}
\end{figure}

\subsection{Noise channels}\label{app:noise}
Here we explicitly give the definitions of the noise channels we consider in the main text, including (a) global depolarizing (white noise), (b) bit flip, (c) phase flip, (d) amplitude damping, and (e) phase damping channels.

For one qubit, the depolarizing channel is also known as white noise, i.e.,
\begin{eqnarray}
    C_{\rm D}(\rho) = (1-p)\rho + p\frac{\openone}{2}\,.
\end{eqnarray}
For a bipartite state, the white noise is the global depolarizing channel.
And under a balance local depolarizing channel, the noisy state will be
\begin{eqnarray}
    C_{\rm GD}(\rho)
    &=& (1-p)^2 \rho
    + p(1-p) \tr_A(\rho)\otimes\frac{\openone}{2} \nonumber\\
    &+& p(1-p) \frac{\openone}{2}\otimes \tr_B(\rho)
    + p^2 \frac{\openone}{2}\otimes\frac{\openone}{2}\,.
\end{eqnarray}
Hence, for a maximally entangled bipartite state, where $\tr_A(\rho)=\tr_B(\rho)=\openone/d$, the global and local depolarizing channels are the same.
We have known the robustness threshold $p = 1/(d+1)$ for the local dimension $d$.

For the bit flip channel and the other channel below, we describe them by Kraus operators, i.e.,
\begin{eqnarray}
    C(\rho) = \sum_i E_i\rho E_i^\dagger\,,
\end{eqnarray}
where Kraus operators need to satisfy $\sum_i E_i E_i^\dagger \leq \openone$.

For the bit flip channel, the Kraus operators are
\begin{eqnarray}
    &&C_{\rm BF}(\rho) = E_0\rho E_0^\dagger + E_1\rho E_1^\dagger\,,\nonumber\\
    &&E_0=\sqrt{1-p}\openone\,,\quad
    E_1=\sqrt{p} X\,.
\end{eqnarray}

The phase flip channel can be described as
\begin{eqnarray}
    &&C_{\rm PF}(\rho) = E_0\rho E_0^\dagger + E_1\rho E_1^\dagger\,,\nonumber\\
    &&E_0=\sqrt{1-p}\openone\,,\quad
    E_1=\sqrt{p} Z\,.
\end{eqnarray}

For the amplitude damping channel, it is the description of energy dissipation, that is, effects due to loss of energy from a quantum system.
For one qubit, it is defined as
\begin{eqnarray}
    &&C_{\rm AD}(\rho) = E_0\rho E_0^\dagger + E_1\rho E_1^\dagger\,,\nonumber\\
    &&E_0=\ketbra{0}+\sqrt{1-\gamma}\ketbra{1}\,,\quad
    E_1=\sqrt{\gamma}\ket{0}\bra{1}\,,
\end{eqnarray}
where $\gamma$ can be thought of as the probability of losing a photon.

The phase damping channel describes the loss of quantum information without loss of energy, which happens when a photon scatters randomly as it travels through a waveguide, or how electronic states in an atom are perturbed upon interacting with distant electrical charges.
The channel for qubit systems is
\begin{eqnarray}
    &&C_{\rm PD}(\rho) = E_0\rho E_0^\dagger + E_1\rho E_1^\dagger\,,\nonumber\\
    &&E_0 = \ketbra{0} + \sqrt{1-\gamma} \ketbra{1}\,,\quad
    E_1 = \sqrt{\gamma}\ketbra{1}\,.
\end{eqnarray}

The numerical results for \cref{fig:channel} are presented below.
\begin{table}[htb]
    \centering
    \renewcommand{\arraystretch}{1.1}
    \setlength{\tabcolsep}{10pt}
    \begin{tabular}{|l|c|c|}
        \hline
        Noise Channel & Noise Strength & Fidelity Threshold \\
        \hline
        White Noise      & 0.6665 + 0.0005 ($\frac23$)   & 50.01\% - 0.04\% (50\%) \\
        Bit Flip         & 0.4999 + 0.0002 ($\frac12$)   & 50.01\% - 0.02\% (50\%) \\
        Phase Flip        & 0.4999 + 0.0002 ($\frac12$)  & 50.01\% - 0.02\% (50\%) \\
        Amplitude Damping & 0.9997 + 0.0003 (1)          & 25.87\% - 0.87\% (25\%) \\
        Phase Damping     & 0.9999 + 0.0001 (1)          & 50.50\% - 0.50\% (50\%) \\
        \hline
    \end{tabular}
    \caption{
    Entanglement robustness thresholds for the Bell state $(\ket{00}+\ket{11})/\sqrt{2}$ under various quantum noise channels: global depolarizing (white noise, GD), bit flip (BF), phase flip (PF), amplitude damping (AD), and phase damping (PD).
    Except for white noise (applied to both qubits), all channels act on one side of the state.
    Error bars represent the certified interval between the entanglement and separability bounds obtained by our numerical algorithm.
    }
    \label{tab:channel}
\end{table}


\begin{thebibliography}{49}%
\makeatletter
\providecommand \@ifxundefined [1]{%
 \@ifx{#1\undefined}
}%
\providecommand \@ifnum [1]{%
 \ifnum #1\expandafter \@firstoftwo
 \else \expandafter \@secondoftwo
 \fi
}%
\providecommand \@ifx [1]{%
 \ifx #1\expandafter \@firstoftwo
 \else \expandafter \@secondoftwo
 \fi
}%
\providecommand \natexlab [1]{#1}%
\providecommand \enquote  [1]{``#1''}%
\providecommand \bibnamefont  [1]{#1}%
\providecommand \bibfnamefont [1]{#1}%
\providecommand \citenamefont [1]{#1}%
\providecommand \href@noop [0]{\@secondoftwo}%
\providecommand \href [0]{\begingroup \@sanitize@url \@href}%
\providecommand \@href[1]{\@@startlink{#1}\@@href}%
\providecommand \@@href[1]{\endgroup#1\@@endlink}%
\providecommand \@sanitize@url [0]{\catcode `\\12\catcode `\$12\catcode
  `\&12\catcode `\#12\catcode `\^12\catcode `\_12\catcode `\%12\relax}%
\providecommand \@@startlink[1]{}%
\providecommand \@@endlink[0]{}%
\providecommand \url  [0]{\begingroup\@sanitize@url \@url }%
\providecommand \@url [1]{\endgroup\@href {#1}{\urlprefix }}%
\providecommand \urlprefix  [0]{URL }%
\providecommand \Eprint [0]{\href }%
\providecommand \doibase [0]{https://doi.org/}%
\providecommand \selectlanguage [0]{\@gobble}%
\providecommand \bibinfo  [0]{\@secondoftwo}%
\providecommand \bibfield  [0]{\@secondoftwo}%
\providecommand \translation [1]{[#1]}%
\providecommand \BibitemOpen [0]{}%
\providecommand \bibitemStop [0]{}%
\providecommand \bibitemNoStop [0]{.\EOS\space}%
\providecommand \EOS [0]{\spacefactor3000\relax}%
\providecommand \BibitemShut  [1]{\csname bibitem#1\endcsname}%
\let\auto@bib@innerbib\@empty
\bibitem [{\citenamefont {Horodecki}\ \emph {et~al.}(2009)\citenamefont
  {Horodecki}, \citenamefont {Horodecki}, \citenamefont {Horodecki},\ and\
  \citenamefont {Horodecki}}]{HHHH09}%
  \BibitemOpen
  \bibfield  {author} {\bibinfo {author} {\bibfnamefont {R.}~\bibnamefont
  {Horodecki}}, \bibinfo {author} {\bibfnamefont {P.}~\bibnamefont
  {Horodecki}}, \bibinfo {author} {\bibfnamefont {M.}~\bibnamefont
  {Horodecki}},\ and\ \bibinfo {author} {\bibfnamefont {K.}~\bibnamefont
  {Horodecki}},\ }\bibfield  {title} {\bibinfo {title} {Quantum entanglement},\
  }\href {https://doi.org/10.1103/RevModPhys.81.865} {\bibfield  {journal}
  {\bibinfo  {journal} {Rev. Mod. Phys.}\ }\textbf {\bibinfo {volume} {81}},\
  \bibinfo {pages} {865} (\bibinfo {year} {2009})}\BibitemShut {NoStop}%
\bibitem [{\citenamefont {G\"uhne}\ and\ \citenamefont {T\'oth}(2009)}]{GT09}%
  \BibitemOpen
  \bibfield  {author} {\bibinfo {author} {\bibfnamefont {O.}~\bibnamefont
  {G\"uhne}}\ and\ \bibinfo {author} {\bibfnamefont {G.}~\bibnamefont
  {T\'oth}},\ }\bibfield  {title} {\bibinfo {title} {Entanglement detection},\
  }\href {https://doi.org/https://doi.org/10.1016/j.physrep.2009.02.004}
  {\bibfield  {journal} {\bibinfo  {journal} {Phys. Rep.}\ }\textbf {\bibinfo
  {volume} {474}},\ \bibinfo {pages} {1} (\bibinfo {year} {2009})}\BibitemShut
  {NoStop}%
\bibitem [{\citenamefont {Peres}(1996)}]{Per96}%
  \BibitemOpen
  \bibfield  {author} {\bibinfo {author} {\bibfnamefont {A.}~\bibnamefont
  {Peres}},\ }\bibfield  {title} {\bibinfo {title} {Separability criterion for
  density matrices},\ }\href {https://doi.org/10.1103/PhysRevLett.77.1413}
  {\bibfield  {journal} {\bibinfo  {journal} {Phys. Rev. Lett.}\ }\textbf
  {\bibinfo {volume} {77}},\ \bibinfo {pages} {1413} (\bibinfo {year}
  {1996})}\BibitemShut {NoStop}%
\bibitem [{\citenamefont {Horodecki}(1997)}]{Horo97}%
  \BibitemOpen
  \bibfield  {author} {\bibinfo {author} {\bibfnamefont {P.}~\bibnamefont
  {Horodecki}},\ }\bibfield  {title} {\bibinfo {title} {Separability criterion
  and inseparable mixed states with positive partial transposition},\ }\href
  {https://doi.org/https://doi.org/10.1016/S0375-9601(97)00416-7} {\bibfield
  {journal} {\bibinfo  {journal} {Phys. Lett. A}\ }\textbf {\bibinfo {volume}
  {232}},\ \bibinfo {pages} {333} (\bibinfo {year} {1997})}\BibitemShut
  {NoStop}%
\bibitem [{\citenamefont {Chen}\ and\ \citenamefont {Wu}(2003)}]{CW03}%
  \BibitemOpen
  \bibfield  {author} {\bibinfo {author} {\bibfnamefont {K.}~\bibnamefont
  {Chen}}\ and\ \bibinfo {author} {\bibfnamefont {L.-A.}\ \bibnamefont {Wu}},\
  }\bibfield  {title} {\bibinfo {title} {A matrix realignment method for
  recognizing entanglement},\ }\href {https://doi.org/10.26421/QIC3.3-1}
  {\bibfield  {journal} {\bibinfo  {journal} {Quantum Info. Comput.}\ }\textbf
  {\bibinfo {volume} {3}},\ \bibinfo {pages} {193} (\bibinfo {year}
  {2003})}\BibitemShut {NoStop}%
\bibitem [{\citenamefont {Rudolph}(2005)}]{Rud05}%
  \BibitemOpen
  \bibfield  {author} {\bibinfo {author} {\bibfnamefont {O.}~\bibnamefont
  {Rudolph}},\ }\bibfield  {title} {\bibinfo {title} {Further results on the
  cross norm criterion for separability},\ }\href
  {https://doi.org/10.1007/s11128-005-5664-1} {\bibfield  {journal} {\bibinfo
  {journal} {Quantum Inf. Proc.}\ }\textbf {\bibinfo {volume} {4}},\ \bibinfo
  {pages} {219} (\bibinfo {year} {2005})}\BibitemShut {NoStop}%
\bibitem [{\citenamefont {Novo}\ \emph {et~al.}(2013)\citenamefont {Novo},
  \citenamefont {Moroder},\ and\ \citenamefont {G\"uhne}}]{NMG13}%
  \BibitemOpen
  \bibfield  {author} {\bibinfo {author} {\bibfnamefont {L.}~\bibnamefont
  {Novo}}, \bibinfo {author} {\bibfnamefont {T.}~\bibnamefont {Moroder}},\ and\
  \bibinfo {author} {\bibfnamefont {O.}~\bibnamefont {G\"uhne}},\ }\bibfield
  {title} {\bibinfo {title} {Genuine multiparticle entanglement of
  permutationally invariant states},\ }\href
  {https://doi.org/10.1103/PhysRevA.88.012305} {\bibfield  {journal} {\bibinfo
  {journal} {Phys. Rev. A}\ }\textbf {\bibinfo {volume} {88}},\ \bibinfo
  {pages} {012305} (\bibinfo {year} {2013})}\BibitemShut {NoStop}%
\bibitem [{\citenamefont {Horodecki}\ \emph {et~al.}(1996)\citenamefont
  {Horodecki}, \citenamefont {Horodecki},\ and\ \citenamefont
  {Horodecki}}]{HHH96}%
  \BibitemOpen
  \bibfield  {author} {\bibinfo {author} {\bibfnamefont {M.}~\bibnamefont
  {Horodecki}}, \bibinfo {author} {\bibfnamefont {P.}~\bibnamefont
  {Horodecki}},\ and\ \bibinfo {author} {\bibfnamefont {R.}~\bibnamefont
  {Horodecki}},\ }\bibfield  {title} {\bibinfo {title} {Separability of mixed
  states: necessary and sufficient conditions},\ }\href
  {https://doi.org/https://doi.org/10.1016/S0375-9601(96)00706-2} {\bibfield
  {journal} {\bibinfo  {journal} {Phys. Lett. A}\ }\textbf {\bibinfo {volume}
  {223}},\ \bibinfo {pages} {1} (\bibinfo {year} {1996})}\BibitemShut {NoStop}%
\bibitem [{\citenamefont {Terhal}(2000)}]{Ter00}%
  \BibitemOpen
  \bibfield  {author} {\bibinfo {author} {\bibfnamefont {B.~M.}\ \bibnamefont
  {Terhal}},\ }\bibfield  {title} {\bibinfo {title} {{B}ell inequalities and
  the separability criterion},\ }\href
  {https://doi.org/10.1016/s0375-9601(00)00401-1} {\bibfield  {journal}
  {\bibinfo  {journal} {Phys. Lett. A}\ }\textbf {\bibinfo {volume} {271}},\
  \bibinfo {pages} {319} (\bibinfo {year} {2000})}\BibitemShut {NoStop}%
\bibitem [{\citenamefont {Lewenstein}\ \emph {et~al.}(2000)\citenamefont
  {Lewenstein}, \citenamefont {Kraus}, \citenamefont {Cirac},\ and\
  \citenamefont {Horodecki}}]{LKCH00}%
  \BibitemOpen
  \bibfield  {author} {\bibinfo {author} {\bibfnamefont {M.}~\bibnamefont
  {Lewenstein}}, \bibinfo {author} {\bibfnamefont {B.}~\bibnamefont {Kraus}},
  \bibinfo {author} {\bibfnamefont {J.~I.}\ \bibnamefont {Cirac}},\ and\
  \bibinfo {author} {\bibfnamefont {P.}~\bibnamefont {Horodecki}},\ }\bibfield
  {title} {\bibinfo {title} {Optimization of entanglement witnesses},\ }\href
  {https://doi.org/10.1103/PhysRevA.62.052310} {\bibfield  {journal} {\bibinfo
  {journal} {Phys. Rev. A}\ }\textbf {\bibinfo {volume} {62}},\ \bibinfo
  {pages} {052310} (\bibinfo {year} {2000})}\BibitemShut {NoStop}%
\bibitem [{\citenamefont {Bruß}\ \emph {et~al.}(2002)\citenamefont {Bruß},
  \citenamefont {Cirac}, \citenamefont {Horodecki}, \citenamefont {Hulpke},
  \citenamefont {Kraus}, \citenamefont {Lewenstein},\ and\ \citenamefont
  {Sanpera}}]{BCH+02}%
  \BibitemOpen
  \bibfield  {author} {\bibinfo {author} {\bibfnamefont {D.}~\bibnamefont
  {Bruß}}, \bibinfo {author} {\bibfnamefont {J.~I.}\ \bibnamefont {Cirac}},
  \bibinfo {author} {\bibfnamefont {P.}~\bibnamefont {Horodecki}}, \bibinfo
  {author} {\bibfnamefont {F.}~\bibnamefont {Hulpke}}, \bibinfo {author}
  {\bibfnamefont {B.}~\bibnamefont {Kraus}}, \bibinfo {author} {\bibfnamefont
  {M.}~\bibnamefont {Lewenstein}},\ and\ \bibinfo {author} {\bibfnamefont
  {A.}~\bibnamefont {Sanpera}},\ }\bibfield  {title} {\bibinfo {title}
  {Reflections upon separability and distillability},\ }\href
  {https://doi.org/10.1080/09500340110105975} {\bibfield  {journal} {\bibinfo
  {journal} {J. Mod. Opt.}\ }\textbf {\bibinfo {volume} {49}},\ \bibinfo
  {pages} {1399} (\bibinfo {year} {2002})}\BibitemShut {NoStop}%
\bibitem [{\citenamefont {Bourennane}\ \emph {et~al.}(2004)\citenamefont
  {Bourennane}, \citenamefont {Eibl}, \citenamefont {Kurtsiefer}, \citenamefont
  {Gaertner}, \citenamefont {Weinfurter}, \citenamefont {G\"uhne},
  \citenamefont {Hyllus}, \citenamefont {Bru\ss{}}, \citenamefont
  {Lewenstein},\ and\ \citenamefont {Sanpera}}]{BEK+04}%
  \BibitemOpen
  \bibfield  {author} {\bibinfo {author} {\bibfnamefont {M.}~\bibnamefont
  {Bourennane}}, \bibinfo {author} {\bibfnamefont {M.}~\bibnamefont {Eibl}},
  \bibinfo {author} {\bibfnamefont {C.}~\bibnamefont {Kurtsiefer}}, \bibinfo
  {author} {\bibfnamefont {S.}~\bibnamefont {Gaertner}}, \bibinfo {author}
  {\bibfnamefont {H.}~\bibnamefont {Weinfurter}}, \bibinfo {author}
  {\bibfnamefont {O.}~\bibnamefont {G\"uhne}}, \bibinfo {author} {\bibfnamefont
  {P.}~\bibnamefont {Hyllus}}, \bibinfo {author} {\bibfnamefont
  {D.}~\bibnamefont {Bru\ss{}}}, \bibinfo {author} {\bibfnamefont
  {M.}~\bibnamefont {Lewenstein}},\ and\ \bibinfo {author} {\bibfnamefont
  {A.}~\bibnamefont {Sanpera}},\ }\bibfield  {title} {\bibinfo {title}
  {Experimental detection of multipartite entanglement using witness
  operators},\ }\href {https://doi.org/10.1103/PhysRevLett.92.087902}
  {\bibfield  {journal} {\bibinfo  {journal} {Phys. Rev. Lett.}\ }\textbf
  {\bibinfo {volume} {92}},\ \bibinfo {pages} {087902} (\bibinfo {year}
  {2004})}\BibitemShut {NoStop}%
\bibitem [{\citenamefont {Doherty}\ \emph {et~al.}(2004)\citenamefont
  {Doherty}, \citenamefont {Parrilo},\ and\ \citenamefont
  {Spedalieri}}]{DPS04}%
  \BibitemOpen
  \bibfield  {author} {\bibinfo {author} {\bibfnamefont {A.~C.}\ \bibnamefont
  {Doherty}}, \bibinfo {author} {\bibfnamefont {P.~A.}\ \bibnamefont
  {Parrilo}},\ and\ \bibinfo {author} {\bibfnamefont {F.~M.}\ \bibnamefont
  {Spedalieri}},\ }\bibfield  {title} {\bibinfo {title} {Complete family of
  separability criteria},\ }\href {https://doi.org/10.1103/PhysRevA.69.022308}
  {\bibfield  {journal} {\bibinfo  {journal} {Phys. Rev. A}\ }\textbf {\bibinfo
  {volume} {69}},\ \bibinfo {pages} {022308} (\bibinfo {year}
  {2004})}\BibitemShut {NoStop}%
\bibitem [{\citenamefont {Navascu\'es}\ \emph {et~al.}(2009)\citenamefont
  {Navascu\'es}, \citenamefont {Owari},\ and\ \citenamefont {Plenio}}]{NOP09}%
  \BibitemOpen
  \bibfield  {author} {\bibinfo {author} {\bibfnamefont {M.}~\bibnamefont
  {Navascu\'es}}, \bibinfo {author} {\bibfnamefont {M.}~\bibnamefont {Owari}},\
  and\ \bibinfo {author} {\bibfnamefont {M.~B.}\ \bibnamefont {Plenio}},\
  }\bibfield  {title} {\bibinfo {title} {Complete criterion for separability
  detection},\ }\href {https://doi.org/10.1103/PhysRevLett.103.160404}
  {\bibfield  {journal} {\bibinfo  {journal} {Phys. Rev. Lett.}\ }\textbf
  {\bibinfo {volume} {103}},\ \bibinfo {pages} {160404} (\bibinfo {year}
  {2009})}\BibitemShut {NoStop}%
\bibitem [{\citenamefont {Ohst}\ \emph {et~al.}(2024)\citenamefont {Ohst},
  \citenamefont {Yu}, \citenamefont {G\"uhne},\ and\ \citenamefont
  {Nguyen}}]{OYGN24}%
  \BibitemOpen
  \bibfield  {author} {\bibinfo {author} {\bibfnamefont {T.-A.}\ \bibnamefont
  {Ohst}}, \bibinfo {author} {\bibfnamefont {X.-D.}\ \bibnamefont {Yu}},
  \bibinfo {author} {\bibfnamefont {O.}~\bibnamefont {G\"uhne}},\ and\ \bibinfo
  {author} {\bibfnamefont {H.~C.}\ \bibnamefont {Nguyen}},\ }\bibfield  {title}
  {\bibinfo {title} {Certifying quantum separability with adaptive polytopes},\
  }\href {https://doi.org/10.21468/SciPostPhys.16.3.063} {\bibfield  {journal}
  {\bibinfo  {journal} {SciPost Phys.}\ }\textbf {\bibinfo {volume} {16}},\
  \bibinfo {pages} {063} (\bibinfo {year} {2024})}\BibitemShut {NoStop}%
\bibitem [{\citenamefont {Song}\ \emph {et~al.}(2017)\citenamefont {Song},
  \citenamefont {Xu}, \citenamefont {Liu}, \citenamefont {Yang}, \citenamefont
  {Zheng}, \citenamefont {Deng}, \citenamefont {Xie}, \citenamefont {Huang},
  \citenamefont {Guo}, \citenamefont {Zhang}, \citenamefont {Zhang},
  \citenamefont {Xu}, \citenamefont {Zheng}, \citenamefont {Zhu}, \citenamefont
  {Wang}, \citenamefont {Chen}, \citenamefont {Lu}, \citenamefont {Han},\ and\
  \citenamefont {Pan}}]{SXL+17}%
  \BibitemOpen
  \bibfield  {author} {\bibinfo {author} {\bibfnamefont {C.}~\bibnamefont
  {Song}}, \bibinfo {author} {\bibfnamefont {K.}~\bibnamefont {Xu}}, \bibinfo
  {author} {\bibfnamefont {W.}~\bibnamefont {Liu}}, \bibinfo {author}
  {\bibfnamefont {C.-p.}\ \bibnamefont {Yang}}, \bibinfo {author}
  {\bibfnamefont {S.-B.}\ \bibnamefont {Zheng}}, \bibinfo {author}
  {\bibfnamefont {H.}~\bibnamefont {Deng}}, \bibinfo {author} {\bibfnamefont
  {Q.}~\bibnamefont {Xie}}, \bibinfo {author} {\bibfnamefont {K.}~\bibnamefont
  {Huang}}, \bibinfo {author} {\bibfnamefont {Q.}~\bibnamefont {Guo}}, \bibinfo
  {author} {\bibfnamefont {L.}~\bibnamefont {Zhang}}, \bibinfo {author}
  {\bibfnamefont {P.}~\bibnamefont {Zhang}}, \bibinfo {author} {\bibfnamefont
  {D.}~\bibnamefont {Xu}}, \bibinfo {author} {\bibfnamefont {D.}~\bibnamefont
  {Zheng}}, \bibinfo {author} {\bibfnamefont {X.}~\bibnamefont {Zhu}}, \bibinfo
  {author} {\bibfnamefont {H.}~\bibnamefont {Wang}}, \bibinfo {author}
  {\bibfnamefont {Y.-A.}\ \bibnamefont {Chen}}, \bibinfo {author}
  {\bibfnamefont {C.-Y.}\ \bibnamefont {Lu}}, \bibinfo {author} {\bibfnamefont
  {S.}~\bibnamefont {Han}},\ and\ \bibinfo {author} {\bibfnamefont {J.-W.}\
  \bibnamefont {Pan}},\ }\bibfield  {title} {\bibinfo {title} {10-qubit
  entanglement and parallel logic operations with a superconducting circuit},\
  }\href {https://doi.org/10.1103/PhysRevLett.119.180511} {\bibfield  {journal}
  {\bibinfo  {journal} {Phys. Rev. Lett.}\ }\textbf {\bibinfo {volume} {119}},\
  \bibinfo {pages} {180511} (\bibinfo {year} {2017})}\BibitemShut {NoStop}%
\bibitem [{\citenamefont {Kampermann}\ \emph {et~al.}(2012)\citenamefont
  {Kampermann}, \citenamefont {G\"uhne}, \citenamefont {Wilmott},\ and\
  \citenamefont {Bru\ss{}}}]{KGWB12}%
  \BibitemOpen
  \bibfield  {author} {\bibinfo {author} {\bibfnamefont {H.}~\bibnamefont
  {Kampermann}}, \bibinfo {author} {\bibfnamefont {O.}~\bibnamefont {G\"uhne}},
  \bibinfo {author} {\bibfnamefont {C.}~\bibnamefont {Wilmott}},\ and\ \bibinfo
  {author} {\bibfnamefont {D.}~\bibnamefont {Bru\ss{}}},\ }\bibfield  {title}
  {\bibinfo {title} {Algorithm for characterizing stochastic local operations
  and classical communication classes of multiparticle entanglement},\ }\href
  {https://doi.org/10.1103/PhysRevA.86.032307} {\bibfield  {journal} {\bibinfo
  {journal} {Phys. Rev. A}\ }\textbf {\bibinfo {volume} {86}},\ \bibinfo
  {pages} {032307} (\bibinfo {year} {2012})}\BibitemShut {NoStop}%
\bibitem [{\citenamefont {Shang}\ and\ \citenamefont {G\"uhne}(2018)}]{SG18}%
  \BibitemOpen
  \bibfield  {author} {\bibinfo {author} {\bibfnamefont {J.}~\bibnamefont
  {Shang}}\ and\ \bibinfo {author} {\bibfnamefont {O.}~\bibnamefont
  {G\"uhne}},\ }\bibfield  {title} {\bibinfo {title} {Convex optimization over
  classes of multiparticle entanglement},\ }\href
  {https://doi.org/10.1103/PhysRevLett.120.050506} {\bibfield  {journal}
  {\bibinfo  {journal} {Phys. Rev. Lett.}\ }\textbf {\bibinfo {volume} {120}},\
  \bibinfo {pages} {050506} (\bibinfo {year} {2018})}\BibitemShut {NoStop}%
\bibitem [{\citenamefont {Wie\'sniak}\ \emph {et~al.}(2020)\citenamefont
  {Wie\'sniak}, \citenamefont {Pandya}, \citenamefont {Sakarya},\ and\
  \citenamefont {Woloncewicz}}]{WPSW20}%
  \BibitemOpen
  \bibfield  {author} {\bibinfo {author} {\bibfnamefont {M.}~\bibnamefont
  {Wie\'sniak}}, \bibinfo {author} {\bibfnamefont {P.}~\bibnamefont {Pandya}},
  \bibinfo {author} {\bibfnamefont {O.}~\bibnamefont {Sakarya}},\ and\ \bibinfo
  {author} {\bibfnamefont {B.}~\bibnamefont {Woloncewicz}},\ }\bibfield
  {title} {\bibinfo {title} {Distance between bound entangled states from
  unextendible product bases and separable states},\ }\href
  {https://doi.org/10.3390/quantum2010004} {\bibfield  {journal} {\bibinfo
  {journal} {Quantum Rep.}\ }\textbf {\bibinfo {volume} {2}},\ \bibinfo {pages}
  {49} (\bibinfo {year} {2020})}\BibitemShut {NoStop}%
\bibitem [{\citenamefont {Pandya}\ \emph {et~al.}(2020)\citenamefont {Pandya},
  \citenamefont {Sakarya},\ and\ \citenamefont {Wie\'sniak}}]{PSW20}%
  \BibitemOpen
  \bibfield  {author} {\bibinfo {author} {\bibfnamefont {P.}~\bibnamefont
  {Pandya}}, \bibinfo {author} {\bibfnamefont {O.}~\bibnamefont {Sakarya}},\
  and\ \bibinfo {author} {\bibfnamefont {M.}~\bibnamefont {Wie\'sniak}},\
  }\bibfield  {title} {\bibinfo {title} {Hilbert-{S}chmidt distance and
  entanglement witnessing},\ }\href
  {https://doi.org/10.1103/PhysRevA.102.012409} {\bibfield  {journal} {\bibinfo
   {journal} {Phys. Rev. A}\ }\textbf {\bibinfo {volume} {102}},\ \bibinfo
  {pages} {012409} (\bibinfo {year} {2020})}\BibitemShut {NoStop}%
\bibitem [{\citenamefont {Hu}\ \emph {et~al.}(2023)\citenamefont {Hu},
  \citenamefont {Liu},\ and\ \citenamefont {Shang}}]{Hu_2023_algorithm}%
  \BibitemOpen
  \bibfield  {author} {\bibinfo {author} {\bibfnamefont {Y.}~\bibnamefont
  {Hu}}, \bibinfo {author} {\bibfnamefont {Y.-C.}\ \bibnamefont {Liu}},\ and\
  \bibinfo {author} {\bibfnamefont {J.}~\bibnamefont {Shang}},\ }\bibfield
  {title} {\bibinfo {title} {Algorithm for evaluating distance-based
  entanglement measures},\ }\href {https://doi.org/10.1088/1674-1056/acd5c5}
  {\bibfield  {journal} {\bibinfo  {journal} {Chinese Phys. B}\ }\textbf
  {\bibinfo {volume} {32}},\ \bibinfo {pages} {080307} (\bibinfo {year}
  {2023})}\BibitemShut {NoStop}%
\bibitem [{\citenamefont {Gurvits}(2003)}]{Gur03}%
  \BibitemOpen
  \bibfield  {author} {\bibinfo {author} {\bibfnamefont {L.}~\bibnamefont
  {Gurvits}},\ }\bibfield  {title} {\bibinfo {title} {Classical deterministic
  complexity of {E}dmonds' problem and quantum entanglement},\ }in\ \href
  {https://doi.org/10.1145/780542.780545} {\emph {\bibinfo {booktitle} {Proc.
  STOC'03}}},\ \bibinfo {series and number} {STOC '03}\ (\bibinfo  {publisher}
  {Association for Computing Machinery},\ \bibinfo {address} {New York, NY,
  USA},\ \bibinfo {year} {2003})\ pp.\ \bibinfo {pages} {10--19}\BibitemShut
  {NoStop}%
\bibitem [{\citenamefont {Gilbert}(1966)}]{Gil66}%
  \BibitemOpen
  \bibfield  {author} {\bibinfo {author} {\bibfnamefont {E.~G.}\ \bibnamefont
  {Gilbert}},\ }\bibfield  {title} {\bibinfo {title} {An iterative procedure
  for computing the minimum of a quadratic form on a convex set},\ }\href
  {https://doi.org/10.1137/0304007} {\bibfield  {journal} {\bibinfo  {journal}
  {J. SIAM Control}\ }\textbf {\bibinfo {volume} {4}},\ \bibinfo {pages} {61}
  (\bibinfo {year} {1966})}\BibitemShut {NoStop}%
\bibitem [{\citenamefont {Frank}\ and\ \citenamefont {Wolfe}(1956)}]{FW56}%
  \BibitemOpen
  \bibfield  {author} {\bibinfo {author} {\bibfnamefont {M.}~\bibnamefont
  {Frank}}\ and\ \bibinfo {author} {\bibfnamefont {P.}~\bibnamefont {Wolfe}},\
  }\bibfield  {title} {\bibinfo {title} {An algorithm for quadratic
  programming},\ }\href
  {https://doi.org/https://doi.org/10.1002/nav.3800030109} {\bibfield
  {journal} {\bibinfo  {journal} {Nav. Res. Logist. Q.}\ }\textbf {\bibinfo
  {volume} {3}},\ \bibinfo {pages} {95} (\bibinfo {year} {1956})}\BibitemShut
  {NoStop}%
\bibitem [{\citenamefont {Levitin}\ and\ \citenamefont {Polyak}(1966)}]{LP66}%
  \BibitemOpen
  \bibfield  {author} {\bibinfo {author} {\bibfnamefont {E.~S.}\ \bibnamefont
  {Levitin}}\ and\ \bibinfo {author} {\bibfnamefont {B.~T.}\ \bibnamefont
  {Polyak}},\ }\bibfield  {title} {\bibinfo {title} {Constrained minimization
  methods},\ }\href {https://doi.org/10.1016/0041-5553(66)90114-5} {\bibfield
  {journal} {\bibinfo  {journal} {USSR Comput. Math. \& Math. Phys.}\ }\textbf
  {\bibinfo {volume} {6}},\ \bibinfo {pages} {1} (\bibinfo {year}
  {1966})}\BibitemShut {NoStop}%
\bibitem [{\citenamefont {Wolfe}(1976)}]{wolfe1976normpoint}%
  \BibitemOpen
  \bibfield  {author} {\bibinfo {author} {\bibfnamefont {P.}~\bibnamefont
  {Wolfe}},\ }\bibfield  {title} {\bibinfo {title} {Finding the nearest point
  in a polytope},\ }\href {https://doi.org/10.1007/BF01580381} {\bibfield
  {journal} {\bibinfo  {journal} {Math. Program.}\ }\textbf {\bibinfo {volume}
  {11}},\ \bibinfo {pages} {128} (\bibinfo {year} {1976})}\BibitemShut
  {NoStop}%
\bibitem [{\citenamefont {Gu{\'e}lat}\ and\ \citenamefont
  {Marcotte}(1986)}]{guelat1986some}%
  \BibitemOpen
  \bibfield  {author} {\bibinfo {author} {\bibfnamefont {J.}~\bibnamefont
  {Gu{\'e}lat}}\ and\ \bibinfo {author} {\bibfnamefont {P.}~\bibnamefont
  {Marcotte}},\ }\bibfield  {title} {\bibinfo {title} {Some comments on
  {Wolfe}'s `away step'},\ }\href@noop {} {\bibfield  {journal} {\bibinfo
  {journal} {Math. Program.}\ }\textbf {\bibinfo {volume} {35}},\ \bibinfo
  {pages} {110} (\bibinfo {year} {1986})}\BibitemShut {NoStop}%
\bibitem [{\citenamefont {Holloway}(1974)}]{holloway1974extension}%
  \BibitemOpen
  \bibfield  {author} {\bibinfo {author} {\bibfnamefont {C.~A.}\ \bibnamefont
  {Holloway}},\ }\bibfield  {title} {\bibinfo {title} {An extension of the
  {Frank} and {Wolfe} method of feasible directions},\ }\href
  {https://doi.org/10.1007/BF01580219} {\bibfield  {journal} {\bibinfo
  {journal} {Math. Program.}\ }\textbf {\bibinfo {volume} {6}},\ \bibinfo
  {pages} {14–27} (\bibinfo {year} {1974})}\BibitemShut {NoStop}%
\bibitem [{\citenamefont {Tsuji}\ \emph {et~al.}(2022)\citenamefont {Tsuji},
  \citenamefont {Tanaka},\ and\ \citenamefont {Pokutta}}]{tsuji2022pairwise}%
  \BibitemOpen
  \bibfield  {author} {\bibinfo {author} {\bibfnamefont {K.~K.}\ \bibnamefont
  {Tsuji}}, \bibinfo {author} {\bibfnamefont {K.}~\bibnamefont {Tanaka}},\ and\
  \bibinfo {author} {\bibfnamefont {S.}~\bibnamefont {Pokutta}},\ }\bibfield
  {title} {\bibinfo {title} {Pairwise conditional gradients without swap steps
  and sparser kernel herding},\ }in\ \href {https://arxiv.org/abs/2110.12650}
  {\emph {\bibinfo {booktitle} {Proc. 39th ICML}}}\ (\bibinfo {organization}
  {PMLR},\ \bibinfo {year} {2022})\ pp.\ \bibinfo {pages}
  {21864--21883}\BibitemShut {NoStop}%
\bibitem [{\citenamefont {Braun}\ \emph {et~al.}(2017)\citenamefont {Braun},
  \citenamefont {Pokutta},\ and\ \citenamefont {Zink}}]{pok17lazy}%
  \BibitemOpen
  \bibfield  {author} {\bibinfo {author} {\bibfnamefont {G.}~\bibnamefont
  {Braun}}, \bibinfo {author} {\bibfnamefont {S.}~\bibnamefont {Pokutta}},\
  and\ \bibinfo {author} {\bibfnamefont {D.}~\bibnamefont {Zink}},\ }\bibfield
  {title} {\bibinfo {title} {Lazifying conditional gradient algorithms},\ }in\
  \href {https://arxiv.org/abs/1610.05120} {\emph {\bibinfo {booktitle} {Proc.
  34th ICML}}}\ (\bibinfo {year} {2017})\ pp.\ \bibinfo {pages}
  {566--575}\BibitemShut {NoStop}%
\bibitem [{\citenamefont {Halbey}\ \emph {et~al.}(2025)\citenamefont {Halbey},
  \citenamefont {Rakotomandimby}, \citenamefont {Besan\c{c}on}, \citenamefont
  {Designolle},\ and\ \citenamefont
  {Pokutta}}]{halbey2025efficientquadraticcorrectionsfrankwolfe}%
  \BibitemOpen
  \bibfield  {author} {\bibinfo {author} {\bibfnamefont {J.}~\bibnamefont
  {Halbey}}, \bibinfo {author} {\bibfnamefont {S.}~\bibnamefont
  {Rakotomandimby}}, \bibinfo {author} {\bibfnamefont {M.}~\bibnamefont
  {Besan\c{c}on}}, \bibinfo {author} {\bibfnamefont {S.}~\bibnamefont
  {Designolle}},\ and\ \bibinfo {author} {\bibfnamefont {S.}~\bibnamefont
  {Pokutta}},\ }\bibfield  {title} {\bibinfo {title} {Efficient quadratic
  corrections for frank-wolfe algorithms},\ }in\ \href
  {https://doi.org/10.52202/085713-0819} {\emph {\bibinfo {booktitle} {Advances
  in Neural Information Processing Systems}}},\ Vol.~\bibinfo {volume} {38},\
  \bibinfo {editor} {edited by\ \bibinfo {editor} {\bibfnamefont
  {D.}~\bibnamefont {Belgrave}}, \bibinfo {editor} {\bibfnamefont
  {C.}~\bibnamefont {Zhang}}, \bibinfo {editor} {\bibfnamefont
  {H.}~\bibnamefont {Lin}}, \bibinfo {editor} {\bibfnamefont {R.}~\bibnamefont
  {Pascanu}}, \bibinfo {editor} {\bibfnamefont {P.}~\bibnamefont {Koniusz}},
  \bibinfo {editor} {\bibfnamefont {M.}~\bibnamefont {Ghassemi}},\ and\
  \bibinfo {editor} {\bibfnamefont {N.}~\bibnamefont {Chen}}}\ (\bibinfo
  {publisher} {Curran Associates, Inc.},\ \bibinfo {year} {2025})\ pp.\
  \bibinfo {pages} {24229--24273}\BibitemShut {NoStop}%
\bibitem [{sup()}]{supp}%
  \BibitemOpen
  \href@noop {} {}\bibinfo {note} {See Supplemental Material for the
  Appendices.}\BibitemShut {Stop}%
\bibitem [{\citenamefont {Bomze}\ \emph {et~al.}(2021)\citenamefont {Bomze},
  \citenamefont {Rinaldi},\ and\ \citenamefont {Zeffiro}}]{BRZ21}%
  \BibitemOpen
  \bibfield  {author} {\bibinfo {author} {\bibfnamefont {I.~M.}\ \bibnamefont
  {Bomze}}, \bibinfo {author} {\bibfnamefont {F.}~\bibnamefont {Rinaldi}},\
  and\ \bibinfo {author} {\bibfnamefont {D.}~\bibnamefont {Zeffiro}},\
  }\bibfield  {title} {\bibinfo {title} {{F}rank-{W}olfe and friends: a journey
  into projection-free first-order optimization methods},\ }\href
  {https://doi.org/10.1007/s10288-021-00493-y} {\bibfield  {journal} {\bibinfo
  {journal} {4OR}\ }\textbf {\bibinfo {volume} {19}},\ \bibinfo {pages} {313}
  (\bibinfo {year} {2021})}\BibitemShut {NoStop}%
\bibitem [{\citenamefont {Braun}\ \emph {et~al.}(2022)\citenamefont {Braun},
  \citenamefont {Carderera}, \citenamefont {Combettes}, \citenamefont
  {Hassani}, \citenamefont {Karbasi}, \citenamefont {Mokhtari},\ and\
  \citenamefont {Pokutta}}]{BCC+22}%
  \BibitemOpen
  \bibfield  {author} {\bibinfo {author} {\bibfnamefont {G.}~\bibnamefont
  {Braun}}, \bibinfo {author} {\bibfnamefont {A.}~\bibnamefont {Carderera}},
  \bibinfo {author} {\bibfnamefont {C.~W.}\ \bibnamefont {Combettes}}, \bibinfo
  {author} {\bibfnamefont {H.}~\bibnamefont {Hassani}}, \bibinfo {author}
  {\bibfnamefont {A.}~\bibnamefont {Karbasi}}, \bibinfo {author} {\bibfnamefont
  {A.}~\bibnamefont {Mokhtari}},\ and\ \bibinfo {author} {\bibfnamefont
  {S.}~\bibnamefont {Pokutta}},\ }\bibfield  {title} {\bibinfo {title}
  {Conditional gradient methods},\ }\href {https://arxiv.org/abs/2211.14103}
  {\bibfield  {journal} {\bibinfo  {journal} {arXiv:2211.14103}\ } (\bibinfo
  {year} {2022})}\BibitemShut {NoStop}%
\bibitem [{Note1()}]{Note1}%
  \BibitemOpen
  \bibinfo {note} {\label {fn:code}See our code as a Julia package at \protect
  \url {https://github.com/ZIB-IOL/EntanglementDetection.jl}.}\BibitemShut
  {Stop}%
\bibitem [{\citenamefont {Gühne}\ and\ \citenamefont
  {Seevinck}(2010)}]{Guhne_2010_separability}%
  \BibitemOpen
  \bibfield  {author} {\bibinfo {author} {\bibfnamefont {O.}~\bibnamefont
  {Gühne}}\ and\ \bibinfo {author} {\bibfnamefont {M.}~\bibnamefont
  {Seevinck}},\ }\bibfield  {title} {\bibinfo {title} {Separability criteria
  for genuine multiparticle entanglement},\ }\href
  {https://doi.org/10.1088/1367-2630/12/5/053002} {\bibfield  {journal}
  {\bibinfo  {journal} {New J. Phys.}\ }\textbf {\bibinfo {volume} {12}},\
  \bibinfo {pages} {053002} (\bibinfo {year} {2010})}\BibitemShut {NoStop}%
\bibitem [{\citenamefont {Gao}\ and\ \citenamefont
  {Hong}(2010)}]{Gao_2010_detection}%
  \BibitemOpen
  \bibfield  {author} {\bibinfo {author} {\bibfnamefont {T.}~\bibnamefont
  {Gao}}\ and\ \bibinfo {author} {\bibfnamefont {Y.}~\bibnamefont {Hong}},\
  }\bibfield  {title} {\bibinfo {title} {Detection of genuinely entangled and
  nonseparable $n$-partite quantum states},\ }\href
  {https://doi.org/10.1103/PhysRevA.82.062113} {\bibfield  {journal} {\bibinfo
  {journal} {Phys. Rev. A}\ }\textbf {\bibinfo {volume} {82}},\ \bibinfo
  {pages} {062113} (\bibinfo {year} {2010})}\BibitemShut {NoStop}%
\bibitem [{\citenamefont {Jungnitsch}\ \emph {et~al.}(2011)\citenamefont
  {Jungnitsch}, \citenamefont {Moroder},\ and\ \citenamefont
  {G\"uhne}}]{Jungnitsch_2011_taming}%
  \BibitemOpen
  \bibfield  {author} {\bibinfo {author} {\bibfnamefont {B.}~\bibnamefont
  {Jungnitsch}}, \bibinfo {author} {\bibfnamefont {T.}~\bibnamefont
  {Moroder}},\ and\ \bibinfo {author} {\bibfnamefont {O.}~\bibnamefont
  {G\"uhne}},\ }\bibfield  {title} {\bibinfo {title} {Taming multiparticle
  entanglement},\ }\href {https://doi.org/10.1103/PhysRevLett.106.190502}
  {\bibfield  {journal} {\bibinfo  {journal} {Phys. Rev. Lett.}\ }\textbf
  {\bibinfo {volume} {106}},\ \bibinfo {pages} {190502} (\bibinfo {year}
  {2011})}\BibitemShut {NoStop}%
\bibitem [{\citenamefont {Ananth}\ and\ \citenamefont
  {Senthilvelan}(2016)}]{Ananth_2016_nonkseparability}%
  \BibitemOpen
  \bibfield  {author} {\bibinfo {author} {\bibfnamefont {N.}~\bibnamefont
  {Ananth}}\ and\ \bibinfo {author} {\bibfnamefont {M.}~\bibnamefont
  {Senthilvelan}},\ }\bibfield  {title} {\bibinfo {title} {On the
  non-$k$-separability of {D}icke class of states and $n$-qudit $w$ states},\
  }\href {https://doi.org/10.1007/s10773-015-2825-6} {\bibfield  {journal}
  {\bibinfo  {journal} {Int. J. Theor. Phys.}\ }\textbf {\bibinfo {volume}
  {55}},\ \bibinfo {pages} {1854} (\bibinfo {year} {2016})}\BibitemShut
  {NoStop}%
\bibitem [{\citenamefont {Chen}\ \emph {et~al.}(2018)\citenamefont {Chen},
  \citenamefont {Jiang},\ and\ \citenamefont {Xu}}]{Chen_2018_necessary}%
  \BibitemOpen
  \bibfield  {author} {\bibinfo {author} {\bibfnamefont {X.-Y.}\ \bibnamefont
  {Chen}}, \bibinfo {author} {\bibfnamefont {L.-Z.}\ \bibnamefont {Jiang}},\
  and\ \bibinfo {author} {\bibfnamefont {Z.-A.}\ \bibnamefont {Xu}},\
  }\bibfield  {title} {\bibinfo {title} {Necessary and sufficient criterion for
  $k$-separability of $n$-qubit noisy {GHZ} states},\ }\href
  {https://doi.org/10.1142/s0219749918500375} {\bibfield  {journal} {\bibinfo
  {journal} {Int. J. Quantum Inf.}\ }\textbf {\bibinfo {volume} {16}},\
  \bibinfo {pages} {1850037} (\bibinfo {year} {2018})}\BibitemShut {NoStop}%
\bibitem [{\citenamefont {Ge}\ \emph {et~al.}(2021)\citenamefont {Ge},
  \citenamefont {Xu}, \citenamefont {Hu}, \citenamefont {Jiang},\ and\
  \citenamefont {Chen}}]{Ge_2021_tripartite}%
  \BibitemOpen
  \bibfield  {author} {\bibinfo {author} {\bibfnamefont {L.-L.}\ \bibnamefont
  {Ge}}, \bibinfo {author} {\bibfnamefont {M.}~\bibnamefont {Xu}}, \bibinfo
  {author} {\bibfnamefont {T.}~\bibnamefont {Hu}}, \bibinfo {author}
  {\bibfnamefont {L.-Z.}\ \bibnamefont {Jiang}},\ and\ \bibinfo {author}
  {\bibfnamefont {X.-Y.}\ \bibnamefont {Chen}},\ }\bibfield  {title} {\bibinfo
  {title} {Tripartite separability of four-qubit {W} and {D}icke mixed state in
  noise environment},\ }\href {https://doi.org/10.1140/epjp/s13360-021-01729-w}
  {\bibfield  {journal} {\bibinfo  {journal} {Eur. Phys. J. Plus}\ }\textbf
  {\bibinfo {volume} {136}},\ \bibinfo {pages} {1} (\bibinfo {year}
  {2021})}\BibitemShut {NoStop}%
\bibitem [{\citenamefont {Pittenger}\ and\ \citenamefont {Rubin}(2001)}]{PR01}%
  \BibitemOpen
  \bibfield  {author} {\bibinfo {author} {\bibfnamefont {A.~O.}\ \bibnamefont
  {Pittenger}}\ and\ \bibinfo {author} {\bibfnamefont {M.~H.}\ \bibnamefont
  {Rubin}},\ }\bibfield  {title} {\bibinfo {title} {Convexity and the
  separability problem of quantum mechanical density matrices},\ }\href
  {https://api.semanticscholar.org/CorpusID:15260929} {\bibfield  {journal}
  {\bibinfo  {journal} {Linear Algebra Appl.}\ }\textbf {\bibinfo {volume}
  {346}},\ \bibinfo {pages} {47} (\bibinfo {year} {2001})}\BibitemShut
  {NoStop}%
\bibitem [{\citenamefont {Bertlmann}\ \emph {et~al.}(2002)\citenamefont
  {Bertlmann}, \citenamefont {Narnhofer},\ and\ \citenamefont
  {Thirring}}]{BNT02}%
  \BibitemOpen
  \bibfield  {author} {\bibinfo {author} {\bibfnamefont {R.~A.}\ \bibnamefont
  {Bertlmann}}, \bibinfo {author} {\bibfnamefont {H.}~\bibnamefont
  {Narnhofer}},\ and\ \bibinfo {author} {\bibfnamefont {W.}~\bibnamefont
  {Thirring}},\ }\bibfield  {title} {\bibinfo {title} {Geometric picture of
  entanglement and {B}ell inequalities},\ }\href
  {https://doi.org/10.1103/PhysRevA.66.032319} {\bibfield  {journal} {\bibinfo
  {journal} {Phys. Rev. A}\ }\textbf {\bibinfo {volume} {66}},\ \bibinfo
  {pages} {032319} (\bibinfo {year} {2002})}\BibitemShut {NoStop}%
\bibitem [{\citenamefont {Bertlmann}\ \emph {et~al.}(2005)\citenamefont
  {Bertlmann}, \citenamefont {Durstberger}, \citenamefont {Hiesmayr},\ and\
  \citenamefont {Krammer}}]{Bertlmann2005optimal}%
  \BibitemOpen
  \bibfield  {author} {\bibinfo {author} {\bibfnamefont {R.~A.}\ \bibnamefont
  {Bertlmann}}, \bibinfo {author} {\bibfnamefont {K.}~\bibnamefont
  {Durstberger}}, \bibinfo {author} {\bibfnamefont {B.~C.}\ \bibnamefont
  {Hiesmayr}},\ and\ \bibinfo {author} {\bibfnamefont {P.}~\bibnamefont
  {Krammer}},\ }\bibfield  {title} {\bibinfo {title} {Optimal entanglement
  witnesses for qubits and qutrits},\ }\href
  {https://doi.org/10.1103/PhysRevA.72.052331} {\bibfield  {journal} {\bibinfo
  {journal} {Phys. Rev. A}\ }\textbf {\bibinfo {volume} {72}},\ \bibinfo
  {pages} {052331} (\bibinfo {year} {2005})}\BibitemShut {NoStop}%
\bibitem [{\citenamefont {Pisier}(1989)}]{Pis89}%
  \BibitemOpen
  \bibfield  {author} {\bibinfo {author} {\bibfnamefont {G.}~\bibnamefont
  {Pisier}},\ }\href {https://doi.org/10.1017/CBO9780511662454} {\emph
  {\bibinfo {title} {The Volume of Convex Bodies}}},\ Cambridge Tracts in
  Mathematics\ (\bibinfo  {publisher} {Cambridge University Press},\ \bibinfo
  {year} {1989})\BibitemShut {NoStop}%
\bibitem [{\citenamefont {Aubrun}\ and\ \citenamefont {Szarek}(2017)}]{AS17}%
  \BibitemOpen
  \bibfield  {author} {\bibinfo {author} {\bibfnamefont {G.}~\bibnamefont
  {Aubrun}}\ and\ \bibinfo {author} {\bibfnamefont {S.~J.}\ \bibnamefont
  {Szarek}},\ }\href {https://doi.org/10.1090/surv/223} {\emph {\bibinfo
  {title} {Alice and Bob meet Banach}}},\ Mathematical Surveys and Monographs,
  Volume 223\ (\bibinfo  {publisher} {American Mathematical Society},\ \bibinfo
  {year} {2017})\BibitemShut {NoStop}%
\bibitem [{\citenamefont {Edelsbrunner}\ and\ \citenamefont
  {Grayson}(2000)}]{EG00}%
  \BibitemOpen
  \bibfield  {author} {\bibinfo {author} {\bibfnamefont {H.}~\bibnamefont
  {Edelsbrunner}}\ and\ \bibinfo {author} {\bibfnamefont {D.~R.}\ \bibnamefont
  {Grayson}},\ }\bibfield  {title} {\bibinfo {title} {Edgewise subdivision of a
  simplex},\ }\href {https://doi.org/10.1007/s004540010063} {\bibfield
  {journal} {\bibinfo  {journal} {Discrete Comput. Geom.}\ }\textbf {\bibinfo
  {volume} {24}},\ \bibinfo {pages} {707} (\bibinfo {year} {2000})}\BibitemShut
  {NoStop}%
\bibitem [{\citenamefont {Gurvits}\ and\ \citenamefont {Barnum}(2003)}]{GB03}%
  \BibitemOpen
  \bibfield  {author} {\bibinfo {author} {\bibfnamefont {L.}~\bibnamefont
  {Gurvits}}\ and\ \bibinfo {author} {\bibfnamefont {H.}~\bibnamefont
  {Barnum}},\ }\bibfield  {title} {\bibinfo {title} {Separable balls around the
  maximally mixed multipartite quantum states},\ }\href
  {https://doi.org/10.1103/PhysRevA.68.042312} {\bibfield  {journal} {\bibinfo
  {journal} {Phys. Rev. A}\ }\textbf {\bibinfo {volume} {68}},\ \bibinfo
  {pages} {042312} (\bibinfo {year} {2003})}\BibitemShut {NoStop}%
\bibitem [{\citenamefont {Chen}\ \emph {et~al.}(2012)\citenamefont {Chen},
  \citenamefont {Ma}, \citenamefont {G\"uhne},\ and\ \citenamefont
  {Severini}}]{Chen_2012_estimating}%
  \BibitemOpen
  \bibfield  {author} {\bibinfo {author} {\bibfnamefont {Z.-H.}\ \bibnamefont
  {Chen}}, \bibinfo {author} {\bibfnamefont {Z.-H.}\ \bibnamefont {Ma}},
  \bibinfo {author} {\bibfnamefont {O.}~\bibnamefont {G\"uhne}},\ and\ \bibinfo
  {author} {\bibfnamefont {S.}~\bibnamefont {Severini}},\ }\bibfield  {title}
  {\bibinfo {title} {Estimating entanglement monotones with a generalization of
  the {W}ootters formula},\ }\href
  {https://doi.org/10.1103/PhysRevLett.109.200503} {\bibfield  {journal}
  {\bibinfo  {journal} {Phys. Rev. Lett.}\ }\textbf {\bibinfo {volume} {109}},\
  \bibinfo {pages} {200503} (\bibinfo {year} {2012})}\BibitemShut {NoStop}%
\end{thebibliography}
\end{document}